\documentclass[twoside,11pt]{article}

\usepackage{natbib}
\usepackage{amsmath, mathrsfs, amssymb, amsthm, fullpage, graphicx, mathtools, bbm, hyperref,xcolor}
\usepackage[ruled,vlined]{algorithm2e}
\newtheorem{theorem}{Theorem}[section]
\newtheorem{lemma}{Lemma}[section]

\newtheorem*{example*}{Example}

\makeatletter
\def\BState{\State\hskip-\ALG@thistlm}
\makeatother

\newtheorem{proposition}{Proposition}[section]
\newtheorem{corollary}{Corollary}[section]
  {
      \theoremstyle{plain}
      
  }

\usepackage{caption}
\usepackage{chngcntr}
\counterwithin{figure}{section}

\newcommand{\mme}[0]{\mathbb{E}}
\newcommand{\mmp}[0]{\mathbb{P}}
\newcommand{\mmr}[0]{\mathbb{R}}
\newcommand{\mmn}[0]{\mathbb{N}}

\newcommand{\bone}[0]{\mathbbm{1}}

\usepackage{multirow,float}

\usepackage{lastpage}

\begin{document}

\title{Conformal Inference for Online Prediction with Arbitrary Distribution Shifts}

\author{Isaac Gibbs\footnote{Department of Statistics, Stanford University} \and Emmanuel Cand\`{e}s\footnote{Departments of Statistics and Mathematics, Stanford University}}


\maketitle

\begin{abstract}
We consider the problem of forming prediction sets in an online setting where the distribution generating the data is allowed to vary over time. Previous approaches to this problem suffer from over-weighting historical data and thus may fail to quickly react to the underlying dynamics. Here we correct this issue and develop a novel procedure with provably small regret over all local time intervals of a given width. We achieve this by modifying the adaptive conformal inference (ACI) algorithm of Gibbs and Cand\`{e}s (2021) to contain an additional step in which the step-size parameter of ACI's gradient descent update is tuned over time. Crucially, this means that unlike ACI, which requires knowledge of the rate of change of the data-generating mechanism, our new procedure is adaptive to both the size and type of the distribution shift. Our methods are highly flexible and can be used in combination with any baseline predictive algorithm that produces point estimates or estimated quantiles of the target without the need for distributional assumptions. We test our techniques on two real-world datasets aimed at predicting stock market volatility and COVID-19 case counts and find that they are robust and adaptive to real-world distribution shifts.
\end{abstract}


\section{Introduction}

We consider a situation in which we observe a data stream $\{(X_t,Y_t)\}_{1 \leq t \leq T}$ generated by a dynamic process in which the distribution of $(X_t,Y_t)$ (and more broadly of subsequences $(X_t,Y_t)$, $\dots$, $(X_{t+s},Y_{t+s})$) is allowed to vary over time. At each time point $t$, our goal is to use the previously observed data $\{(X_s,Y_s)\}_{s<t}$, along with the new covariates $X_t$, to form a prediction set for the target value $Y_t$. We are motivated by numerous modern applications in which a complex model (e.g.~neural network, random forest) is employed to produce a point estimate of $Y$. While these models have been found to perform well on i.i.d.~training and testing data, rigorous guarantees on their accuracy are lacking and their empirical performance has been found to degrade under distribution shift (\cite{WILDS2021}). Thus, a more robust understanding of the uncertainty underlying these methods' predictions is necessary before they can be deployed in practice.

Prediction sets have become a popular tool for quantifying the accuracy of machine learning models. Formally, we say that $\hat{C}(\cdot) \subseteq \mmr$ is a $1-\alpha$ prediction set for $Y$ if $\mmp(Y \in \hat{C}(X)) = 1-\alpha$. Conceptually, by examining the size and scope of $\hat{C}(X)$ the user can gain information above the uncertainty underlying a model's point-prediction of $Y$. 

Many of the most useful tools for computing prediction sets come from the field of conformal inference. This general framework provides a flexible set of methodologies for transforming the point or quantile-estimates output by a black-box machine learning model into valid prediction sets \cite[see e.g.][]{Saunders1999, Vovk1999, VovkBook, Gammerman2007, Shafer2008, Lei2014, Sadinle2019, Barber2020, Baber2021}. Conformal inference is particularly powerful because it allows users to leverage improvements in the baseline predictor to obtain smaller and more accurate prediction sets (\cite{Romano2019}). 

The original conformal inference methods developed by Vovk and colleagues typically require that all the training and testing data be exchangeable (e.g.~be i.i.d.), and in particular, require that all points have the same marginal distribution. While the earlier literature does contain some extensions beyond exchangeabiltiy to data sequences that are locally exchangeable or can be transformed to an exchangeable sequence (\cite{VovkBook}), the applicability of these methods is limited. More recently, many authors have extended conformal inference to account for a wider variety of distribution shifts and dependency structures within the training and testing data. Some examples including methods for stationary time series  (\cite{Chernozhukov2018}), cross-sectional time series (\cite{Lin2022}), label shift (\cite{Podkopaev2021}), covariate shift (\cite{Tibs2019, Yang2022}), and generic methods for re-weighting non-identically distributed data (\cite{Barber2022}). 

The problem of adjusting a conformal predictor to adapt to arbitrary online distribution shifts was originally proposed by \citet{Gibbs2021}. They gave a gradient descent method, called adaptive conformal inference (ACI), that tunes the width of the prediction sets to adapt to the underlying uncertainty in the environment. While that method was found to produce good results both theoretically and empirically, its performance critically relies on a good specification of its step-size parameter. Specifically, it was shown that for optimal performance, the step-size should be set proportional to the underlying rate of change in the environment, which is unknown in practice. 

Recently, two alternatives to ACI have been proposed that avoid the need for a user-specified step-size. The most direct approach is that of \citet{Zaffran2022}, which gives an expert learning method for adaptively tuning the step-size based off of the historical performance of a set of candidate values. Moving away from gradient descent based methods, \citet{Bastani2022} propose an alternative approach in which the width of the prediction set is chosen directly from a set of candidate thresholds. While these methods can be argued to improve on ACI, both approaches have the shortcoming of heavily weighting older historical data when choosing amongst their candidate values. As our experiments in Section \ref{sec:simulated_data} show, this can lead to a failure to quickly adapt when abrupt changes occur. 

In this article, we propose an alternative expert selection scheme for choosing the step-size in ACI. We show that unlike previous approaches, our method can control the deviation in the coverage probability locally over time and we bound the local coverage of our method in terms of the local rate of change of an underlying optimal target parameter. The importance of this result is not purely theoretical, and we provide example settings where alternative methods produce worse adaptivity to the local dynamics than our approach. We evaluate the performance of our method on two real-world prediction tasks aimed towards predicting stock market volatility and COVID-19 case counts, and find that it adapts well to real-world dynamics.

\section{Methodology}\label{sec:methods}

\subsection{Conformal inference}\label{sec:basic_conf}

Let $(X_1,Y_1),\dots,(X_n,Y_n) \in \mmr^d \times \mmr$ denote a set of observed training data and $(X_{n+1},Y_{n+1})$ denote a new test point from which we only observe $X_{n+1}$. In order to construct a prediction set, conformal inference begins by imputing guesses $y$ for $Y_{n+1}$. Then, for each candidate value $y$, a conformity score $S : (\mmr^d \times \mmr)^n \times \mmr^d \times \mmr \to \mmr$ is used to measure how well the data point $(X_{n+1},y)$ conforms with $(X_1,Y_1),\dots,(X_n,Y_n)$. Typically, this is done by first using all $n+1$ datapoints to fit a regression and then measuring how well $y$ aligns with the prediction of the fitted model at $X_{n+1}$. For example, we may take $S(\cdot)$ to be the absolute residual
\begin{equation}\label{eq:conf_score_ex}
S((X_j,Y_j)_{1 \leq j \leq n},(X_{n+1},y)) := |y - \hat{\mu}(X_{n+1})|,
\end{equation}
where $\hat{\mu}$ is an estimate of $\mme[Y|X]$ fit on $(X_1,Y_1),\dots,(X_n,Y_n),(X_{n+1},y)$, or the estimated probability
\[
S((X_j,Y_j)_{1 \leq j \leq n},(X_{n+1},y)) := 1 - \hat{\pi}(y|X_{n+1}),
\]
where $\hat{\pi}(y|X_{n+1}) $ is a fitted estimate of $\mmp(Y_{n+1} = y | X_{n+1})$. In the final step of conformal inference the value $y$ is added to the prediction set if the test score, $S((X_j,Y_j)_{1 \leq j \leq n},(X_{n+1},y))$, is small relative to the training scores, $\{S((X_j,Y_j)_{1 \leq j \leq  n, j \neq i},(X_{n+1},y), (X_i,Y_i))\}_{i=1}^n$, e.g. if the residual $|y-\hat{\mu}(X_{n+1})|$ is small relative to $|Y_i - \hat{\mu}(X_i)|$.

In general, the only requirement on $S(\cdot)$ necessary for existing theoretical results to hold is that it is unchanged by permutations of its $n+1$ arguments. In the context of this article, the data $(X_1,Y_1)$,$\dots$,$(X_n,Y_n)$, $(X_{n+1},Y_{n+1})$ will often have a temporal dependence structure. As a result, it may not be sensible to treat all of the arguments to $S(\cdot)$ symmetrically and we will often use conformity scores that are not permutation invariant.

For ease of notation, let $S^y_i := S((X_j,Y_j)_{1 \leq j \leq n, j \neq i},(X_{n+1},y),   (X_i,Y_i))$ and $S^{y}_{n+1}$ denote the test score $S((X_j,Y_j)_{1 \leq j \leq n},(X_{n+1},y))$. For any $\tau \in \mmr$ and distribution $\mathcal{D}$, let $ \text{Quantile}\left(\tau,\mathcal{D} \right)$ denote the $\tau_{\text{th}}$ quantile of $\mathcal{D}$ with the convention that $\text{Quantile}\left(\tau,\mathcal{D} \right) = \infty$ (respectively $-\infty$) for all $\tau \geq 1$ (respectively $\tau \leq 0$). Then, formally conformal inference outputs the prediction set
\begin{equation}\label{eq:conformal_pred_set}
\hat{C}_{n+1} := \left\{y : S_{n+1}^y \leq \text{Quantile}\left(1-\alpha, \frac{1}{n+1} \sum_{i=1}^{n+1} \delta_{S_i^y} \right) \right\}.
\end{equation}
As alluded to in the introduction, this set satisfies the following coverage guarantee.
\begin{theorem}
If the data $(X_1,Y_1),\dots,(X_{n+1},Y_{n+1})$ are exchangeable and $S(\cdot)$ is invariant to permutations of its $n+1$ arguments, then
\[
\mmp(Y_{n+1} \in \hat{C}_{n+1})  \geq 1-\alpha.
\]
Moreover, if in addition the values $\{S_i^{Y_{n+1}}\}_{1 \leq i \leq n+1}$ are distinct with probability one, then
\[
 \mmp(Y_{n+1} \in \hat{C}_{n+1}) \leq 1 - \alpha +\frac{1}{n+1}.
\]
\end{theorem}
A full proof of this result can be found in Lemma 1 of \citet{Romano2019}. For an earlier treatment of the first part of the Theorem see also \citet{VovkBook}.

Unfortunately, in most practical examples, computing $\hat{C}_{n+1}$ requires the user to fit the regression function $\hat{\mu}$ or $\hat{\pi}$ for all possible values of $y$. As this is not usually computationally feasible, many implementations of conformal inference use a data splitting approach in which the regression estimate is pre-fit in advance using a separate set of training data (\cite{Papadopoulos2002, VovkBook, Papadopoulos2008SplitConf}). The methods developed in this article do not rely on any formal guarantee of conformal inference and can be used in conjunction with any procedure that produces estimated quantiles of a conformity score. Thus, in our experiments, we avoid extraneous computation by using procedures that do not strictly adhere to the construction given by (\ref{eq:conformal_pred_set}).

\subsection{Adaptive conformal inference} 

The methodology developed in this article builds upon the adaptive conformal inference (ACI) algorithm proposed by \cite{Gibbs2021}. This procedure accounts for non-exchangeability by treating the quantile of the conformity scores as a tunable parameter that can be learned in an online fashion. More concretely, let $\mathcal{D}^y_t$ denote our estimate of the conformity score distribution at time-step $t$ with imputed value $y$. For instance, in our experiments we will often use the empirical distribution of the most recent $r$ conformity scores, $\mathcal{D}^y_t = \frac{1}{r} \sum_{i=t-r+1}^{t}\delta_{S_i^y}$ (recall that standard conformal inference would take $r=t$). Let
\begin{equation}\label{eq:beta_pred_set}
\hat{C}_{t}(\beta) := \left\{y : S_{t}^y \leq \text{Quantile}\left(1-\beta, \mathcal{D}^y_t \right) \right\},
\end{equation}
denote the prediction set obtained at timestep $t$ using the $1-\beta$ quantile of $\mathcal{D}^y_t$. Then, without any assumptions on the data generating distribution, we know that $\beta \mapsto \mmp(Y_{t} \in \hat{C}_t(\beta))$ is non-increasing with $\mmp(Y_{t} \in \hat{C}_t(0)) = 1$ and $\mmp(Y_{t} \in \hat{C}_t(1)) = 0$. If, in addition, we make the mild assumption that $\beta \mapsto \mmp(Y_{t} \in \hat{C}_t(\beta))$ is continuous, then there exists an optimal value, $\alpha^*_t$, such that $\mmp(Y_{t} \in \hat{C}_t(\alpha^*_t)) = 1-\alpha$. Since we do not know this optimal value \textit{a priori}, ACI estimates it in an online fashion using a parameter $\alpha_t$ that is updated as
\begin{equation}\label{eq:aci_update}
\alpha_{t+1} = \alpha_{t} + \gamma(\alpha - \text{err}_{t}),
\end{equation}
where $\gamma > 0$ is a step-size parameter and
\[
\text{err}_{t} := \begin{cases}
0, \text{ if } Y_{t} \in \hat{C}_{t}(\alpha_{t}),\\
1, \text{ if } Y_{t} \notin \hat{C}_{t}(\alpha_{t}).
\end{cases}
\]
In simple terms, the update (\ref{eq:aci_update}) can be seen as increasing/decreasing the size of the prediction set in response to the historical under/over coverage of the algorithm. 

A natural criticism of this approach, originally raised in \cite{Bastani2022}, is that ACI can obtain good coverage not because it successfully learns $\alpha^*_t$, but rather simply due to the fact that it reactively corrects its past mistakes. In particular, it may be the case that $\alpha_t$ oscillates between being well below and well above $\alpha^*_t$ and thus good coverage is obtained only through a cancellation of positive and negative errors. In Section \ref{sec:cond_on_alphat_cov} we give empirical evidence indicating that the new methods developed in this article do not exhibit such pathological behaviour.

Returning to (\ref{eq:aci_update}), the critical difficulty in implementing ACI is the choice of $\gamma$. \cite{Gibbs2021} give theoretical results suggesting that $\gamma$ should be chosen proportional to the size of the variation in $\alpha^*_t$ across time. However, this value is unknown and they give no procedure for estimating it. Additionally, much of the theory given in \cite{Gibbs2021} is only valid under two additional assumptions.
\begin{enumerate}
\item
The conformity score function $S((X_s,Y_s)_{s < t},(X_{t},y)) = S(X_t,y)$ is a fixed function that depends only on the new data point $(X_t,y)$ and does not use the most recent data $(X_s,Y_s)_{s < t}$ to recalibrate its predictions. For example, under this assumption, the conformity score (\ref{eq:conf_score_ex}) would use a regression function $\hat{\mu}(\cdot)$ that is fixed in advance and not updated as time progresses.
 \item 
 Instead of using an adaptive distribution, $\mathcal{D}_t^y$ to generate quantiles for the prediction set in (\ref{eq:beta_pred_set}) we instead have some fixed reference distribution, $\mathcal{D}$ that the conformity scores are compared to, i.e. the prediction set can be written as
 \[
 \hat{C}_t(\alpha_t) = \left\{ y : S(X_t,y) \leq \text{Quantile}(1-\alpha_t,\mathcal{D}) \right\}.
 \]
\end{enumerate}
These two assumptions are clearly problematic since under distribution shift the most recent data should be used to recalibrate both the regression function and the estimated distribution of the scores. The most obvious consequence of using fixed models is an increase in the size of the prediction sets over time as the true model drifts and the errors in the point predictions made by the regression model grow. More subtly, holding $\mathcal{D}$ fixed can cause the coverage probability, $\mmp(Y_t \in \hat{C}_t(\alpha_t))$, to sharply deviate from $1-\alpha$ as the past conformity scores no longer reflect the current situation. This will lead to large oscillations in $\alpha^*_t$ that increase the difficulty of the online learning problem. In the following sections, we develop a new method and novel theoretical results that make no assumptions on $S(\cdot)$ and $\mathcal{D}$ and thus allow these quantities to be updated over time.

\subsection{Dynamically-tuned adaptive conformal inference}\label{sec:FACI}

In order to describe our new method, it is useful to first observe that the ACI update (\ref{eq:aci_update}) can be viewed as a gradient descent step with respect to the pinball loss. To see this, let
\[
\beta_t := \sup \{\beta : Y_t \in \hat{C}_t(\beta)\},
\]
be the value of $\beta$ such that $\hat{C}_t(\beta_t)$ is the smallest prediction set containing $Y_t$. Recall the definition of the pinball loss
\[
\ell(\beta_t ,\theta ) := \alpha(\beta_t - \theta) - \min\{0,\beta_t - \theta\}.
\]
Then, one can verify that (\ref{eq:aci_update}) is equivalent to the update 
\[
\alpha_{t+1} := \alpha_{t}  - \gamma \nabla_{\theta} \ell(\beta_t,\alpha_t). \footnote{Here we have ignored the edge case $\beta_t = \alpha_t$. In this case to match the original ACI update one should take the smallest subgradient of $\ell(\beta_t,\alpha_t)$.}
\]
Through this lens, ACI can be viewed as a gradient descent procedure with respect to the sequence of convex losses $\{\ell(\beta_t,\cdot)\}$. Thus, in order to learn $\gamma$ we can utilise popular methods from the online convex optimization literature. In particular, we will employ an exponential re-weighting scheme that chooses a value for $\gamma$ based off of the historical performance of a set of candidate values. Methods of this type have a long history dating back to the original work of \citet{Vovk1990}. Our specific procedure is a small modification of an algorithm proposed by \citet{Gradu2021}. The only difference between our approach and that taken by \citet{Gradu2021}, is that their method is designed to work in a control setting in which actions affect the future states of the system. This leads them to consider a surrogate loss function that accounts for the long-term dependence structure induced by the actions. Because we have no such dependence, we consider a simplified version of their method here.

We refer to the resulting procedure as dynamically-tuned adaptive conformal inference (DtACI, Algorithm \ref{alg:faci}). This algorithm takes as input a candidate set of values for $\gamma$ and constructs a corresponding candidate set of values for $\alpha_t$ by running multiple versions of ACI in parallel. In the convex optimization literature these parallel sequences are typically referred to as experts. The final value of $\alpha_t$ output at time $t$ is then chosen from among these experts by evaluating their historical performance. In effect, we learn the optimal value of $\gamma$ in an online fashion, enabling dynamic calibration of the prediction set to the size of the distribution shift in the environment.

\begin{algorithm}
 \KwData{Observed values $\{\beta_t\}_{1 \leq t \leq T}$, set of candidate $\gamma$ values $\{\gamma_i\}_{1 \leq i \leq k}$, starting points $\{\alpha_1^i\}_{1 \leq i \leq k}$, and parameters $\sigma$ and $\eta$.}
 $w^i_1 \leftarrow 1,\ 1 \leq i \leq \in k$\;
 \For{$t=1,2,\dots,T$}{
 	Define the probabilities $p_t^i := w^i_t/\sum_{1 \leq j \leq k}w^{j}_t$, $\forall 1 \leq i \leq k$\;
 	Output $\alpha_t =  \alpha^i_t$ with probability $p_t^i$\;
 	$\bar{w}^i_t \leftarrow w^i_t \exp(-\eta \ell(\beta_t,\alpha_t^i)),\ \forall1 \leq i \leq k$\;
 	$\bar{W}_t \leftarrow \sum_{1 \leq i \leq k} \bar{w}^i_t $\;
 	$w^i_{t+1} \leftarrow (1-\sigma) \bar{w}^i_t + \bar{W}_t \sigma/k$\;
 	$\text{err}^i_{t} := \bone\{Y_t \notin \hat{C}_t(\alpha^i_t)\}$, $\forall 1 \leq i \leq k$\;
 	$\text{err}_{t} := \bone\{Y_t \notin \hat{C}_t(\alpha_t)\}$\;
 	$\alpha^i_{t+1} = \alpha^i_{t} + \gamma_i(\alpha - \text{err}^i_{t})$, $\forall 1 \leq i \leq k$\;
 }
 
 \caption{DtACI, modified version of Algorithm 1 in \cite{Gradu2021}.}
 \label{alg:faci}
\end{algorithm}

Before moving on we note that Algorithm \ref{alg:faci} is not completely parameter free. In fact, while we have removed the need for an unknown step-size parameter, this has come at the cost of adding two unknown weight parameters, $\eta$ and $\sigma$. While this may initially appear to be problematic, in Section \ref{sec:theory}, we will outline a simple procedure for choosing $\sigma$ and $\eta$ that does not involve any unknown quantities. This contrasts sharply with the situation for $\gamma$ in which an optimal choice requires an in-depth knowledge of the distribution shift. Moreover, in some environments the size of the distribution shift can vary over time and a single constant value for $\gamma$ can perform poorly. For example, in Section \ref{sec:simulated_data} we demonstrate a setting in which adaptively tuning $\gamma$ allows us to quickly respond to an abrupt change in the environment, while a more stationary choice of $\gamma$ lags behind. This issue does not exist for $\sigma$ and $\eta$ and we provide extensive empirical evidence demonstrating that a single choice of these parameters performs well across a large variety of environments. Finally, a more theoretical discussion on the optimal settings for $\sigma$ and $\eta$ that justifies a specific fixed choice for these parameters can be found at the end of Section \ref{sec:reg_bounds}.
 
\subsection{Comparison to existing methods}\label{sec:prior_work}
 
As mentioned in the introduction, two other alternatives to ACI have been proposed in the literature. Most closely related to the present paper is the AgACI method of \citet{Zaffran2022}, which aims to learn the value of $\gamma$ in ACI using the adaptive Bernstein online expert aggregation scheme of \citet{Wintenberger2017}. To describe this method, let $\{L_t^i\}$ and $\{\eta_t^i\}$ denote the cumulative loss and learning rate for expert $i \in \{1,\dots,k\}$ given by the recursive updates, $L_t^i = 0$, $\eta_{t}^i = 0$, and
\begin{align*}
    & \ell_{t}^i := (\text{err}_{t-1}-\alpha)(\alpha_t^i - \alpha_t),\\
    & L_t^i := L_{t-1}^i + \frac{1}{2}(\ell_t^i(1+\eta_{t-1}^i\ell_t^i) + 2^{\lceil \log_2(\max_{1 \leq s \leq t}|\ell^i_s|)  \rceil + 1} \bone\{\eta_{t-1}^i\ell_t^i > 1/2\}),\\
    & \eta^i_t := \max\left\{ 2^{-\lceil \log_2(\max_{1 \leq s \leq t}|\ell^i_s|)  \rceil - 1}, \sqrt{\frac{\log(1/k)}{\sum_{s=1}^{t}(\ell_s^i)^2}} \right\}.
\end{align*} 
The details of these definitions are not critical. The most important thing to note is that $\ell_t^i$ is exactly equal to the first order linearization of the difference $\ell(\beta_{t-1},\alpha_{t-1}^i) - \ell(\beta_{t-1},\alpha_t)$ and thus $L_t^i$ can be interpreted as a cumulative loss over time. Then, with these definitions in hand AgACI defines the probabilities 
\[
\widetilde{p}_t^i := \frac{\eta_t^i \exp(-\eta_t^iL_t^i)}{\sum_{j=1}^k\eta_t^j \exp(-\eta_t^jL_t^j)},
\]
and outputs the estimate $\alpha_t := \sum_{i=1}^k \widetilde{p}_t^i \alpha_t^i$.

The primary difference between AgACI and our method is the relative weight given to the historical performance of the experts. To see this, we first observe that by unravelling the DtACI updates, the DtACI weights can be re-written as a mixture distribution where element $s$ considers the most recent $s$ losses. More precisely, we have
\[
w_{t+1}^i = \sum_{s=0}^t(1-\sigma)^{t-s}\bar{W}_{s} \left(\frac{\sigma}{k}\right)^{\bone\{s\neq 0\}}\exp\left(-\eta \sum_{j=s+1}^{t} \ell(\beta_j,\alpha_j^i)\right),
\]
where for ease of notation we have set $\bar{W}_0 = 0$. Without any formal analysis, it can be immediately seen that the more recent datapoints appear more often in this mixture and thus contribute more to our choice of weights. On the other hand, the AgACI weights are based off a cumulative sum of all previous losses and thus assign a similar degree of importance to all historical data-points. The upside of this choice is that in environments where the rate of distribution shift is constant, AgACI can effectively converge on a single optimal step-size. However, this comes at the cost of reduced adaptivity over time. For instance, if the environment starts in a state of slow distribution drift, but then undergoes an abrupt shift, AgACI can fail to increase the step-size quickly and thus be slow to react to the change. Empirical examples demonstrating these properties are given in Section \ref{sec:simulated_data}. 

The second alternative to ACI that we consider is the multivalid conformal prediction (MVP) method of \citet{Bastani2022}. Instead of targeting the optimal parameter $\alpha^*_t$, this algorithm chooses $\alpha_t$ in order to explicitly obtain the desired long-term miscoverage $\lim_{T \to \infty} T^{-1} \sum_{t=1}^T \text{err}_t = \alpha$. In addition, MVP is also designed to satisfy threshold-calibrated coverage, i.e. for every $\tau$, $\lim_{T \to \infty} (\sum_{t=1}^T\bone\{\alpha_t = \tau\})^{-1} \sum_{t=1}^T \text{err}_t\bone\{\alpha_t = \tau\} = \alpha$. At a high level, this is accomplished by setting a grid of possible choices for $\alpha_t$, and then at each time step outputting the value in the grid that has produced the best historical coverage. 

While MVP can perform well in stationary environments where there exists a single optimal choice for the threshold, it does not give significant adaptivity to local changes. This is demonstrated by our experiments in Section \ref{sec:simulated_data} where MVP fails to adjust to the local variation in $\alpha_t^*$. For a more complete description of the MVP algorithm see Section \ref{sec:MVP} of the Appendix.

Finally, we emphasize that the good local coverage properties of DtACI are not solely an empirical phenomena. Indeed, in the next section we give bounds that control the difference between the estimates $\alpha_t$ produced by DtACI and the optimal values, $\alpha^*_t$ over any local time interval. This theory is new to our methods and no similar results exist for AgACI or MVP.

\section{Coverage properties of DtACI}\label{sec:theory}

In this section we outline the main coverage guarantees of DtACI. We begin by drawing from known results in the online convex optimization literature that bound a quantity known as the dynamic regret of DtACI. We then draw a connection between this regret and the coverage. Finally, we evaluate the long-term coverage in a specialized case where the hyperparameters $\eta$ and $\sigma$ decay to $0$ over time. All proofs are deferred to the Appendix.

\subsection{Dynamic regret of DtACI} \label{sec:reg_bounds}

Our first result quantifies the error in the expert aggregation scheme by bounding the difference between the loss we obtain and that of the best expert.

\begin{lemma}[Modified version of Lemma A.2 in \cite{Gradu2021}]\label{lem:adapt_reg_bound}
Assume that $\sigma \leq 1/2$. Then, for any interval $I = [r,s] \subseteq [T]$ and any $1 \leq i \leq k$,
\begin{equation}\label{eq:expert_reg}
 \sum_{t=r}^s \mme[\ell(\beta_t,\alpha_t)] \leq   \sum_{t=r}^s \ell(\beta_t,\alpha_t^i) + \eta\sum_{t=r}^s \mme[\ell(\beta_t,\alpha_t)^2]  +  \frac{1}{\eta}\left(\log(k/\sigma\right) + |I|2\sigma ),
\end{equation}
where the expectation is over the randomness in Algorithm \ref{alg:faci} and the data $\beta_1,\dots,\beta_T$ can be viewed as fixed.
\end{lemma}

With this lemma in hand, we now turn to our true target, namely the values $\alpha^*_t$ defined in Section \ref{sec:FACI}. Our first step is to recall the following regret bound for gradient descent with a dynamic target.

\begin{lemma}[Application of Theorem 10.1 of \citet{Hazan2019}]\label{lem:reg_ogd}
For any fixed interval $I = [r,s]$, sequence $\alpha^*_r,\dots,\alpha^*_s$ , and $1 \leq i \leq k$,
\[
\sum_{t=r}^s \ell(\beta_t,\alpha_t^i) - \sum_{t=r}^s \ell(\beta_t,\alpha^*_t) \leq \frac{3}{2\gamma_i} (1+\gamma_i)^2 \left(\sum_{t=r+1}^s |\alpha^*_t - \alpha^*_{t-1}| + 1 \right) + \frac{1}{2}\gamma_i |I|.
\]
\end{lemma}

By combining the previous two lemmas we obtain the main result of this section. 

\begin{theorem}\label{thm:dynamic_reg}
Let $\gamma_{\textnormal{max}} := \max_{1 \leq i \leq k} \gamma_i$ and assume that $\gamma_1 < \gamma_2 < \cdots < \gamma_{k}$ with $\gamma_{i+1}/\gamma_i \leq 2$ for all $1 < i \leq k$. Assume additionally that $\gamma_{k} \geq \sqrt{1+1/|I|}$ and $\sigma \leq 1/2$. Then, for any interval $I = [r,s] \subseteq [T]$ and any sequence $\alpha^*_r, \dots, \alpha^*_s \in [0,1]$,
\begin{align*}
 \frac{1}{|I|} \sum_{t=r}^s\mme[ \ell(\beta_t,\alpha_t)] - & \frac{1}{|I|} \sum_{t=r}^s \ell(\beta_t,\alpha_t^*) \leq  \frac{ \log(k/\sigma) + 2\sigma|I| }{\eta|I|} + \frac{\eta}{|I|}\sum_{t=r}^s \mme[\ell(\beta_t,\alpha_t)^2]\\
 & \ \ \ \ \ \ \ \ \ \ \ \ \ \ \ \  + 4(1+\gamma_{\textnormal{max}})^2 \max\left\{\sqrt{ \frac{ \sum_{t=r+1}^s |\alpha^*_t - \alpha^*_{t-1}|+1}{|I|}}, \gamma_1\right\} ,
\end{align*}
where the expectation is over the randomness in Algorithm \ref{alg:faci} and the data $\beta_1,\dots,\beta_T$ can be viewed as fixed.
\end{theorem} 

If we assume that $\gamma_1 \leq \sqrt{ \frac{ \sum_{t=r+1}^s |\alpha^*_t - \alpha^*_{t-1}|+1}{|I|}}$ and take the optimal choices $\sigma = 1/(2|I|)$ and $\eta = \sqrt{\frac{\log(k\cdot |I|) + 2}{\sum_{t=r}^s \mme[\ell(\beta_t,\alpha_t)^2]}}$, we obtain the much simpler bound,
\begin{align*}
 \frac{1}{|I|} \sum_{t=r}^s \mme[\ell(\beta_t,\alpha_t)] - \frac{1}{|I|} \sum_{t=r}^s \ell(\beta_t,\alpha_t^*) & \leq 2 \sqrt{\frac{\log(k\cdot |I|) + 2}{|I|}} \sqrt{\frac{1}{|I|} \sum_{t=r}^s \mme[\ell(\beta_t,\alpha_t)^2]}\\
 & \ \ \ \ +  4(1+\gamma_{\text{max}})^2\sqrt{ \frac{ \sum_{t=r+1}^s |\alpha^*_t - \alpha^*_{t-1}|+1}{|I|}} \\
 & = O\left(\sqrt{\frac{\log(|I|)}{|I|}} \right) + O\left(\sqrt{ \frac{ \sum_{t=r+1}^s |\alpha^*_t - \alpha^*_{t-1}|}{|I|}} \right).
\end{align*}

The quantity $\frac{ \sum_{t=r+1}^s |\alpha^*_t - \alpha^*_{t-1}|}{|I|}$ can be viewed as a one-dimensional quantification of the size of the distribution shift in the environment. Thus, Theorem \ref{thm:dynamic_reg} gives a direct control on the average peformance of DtACI in terms of the distribution shift. We emphasize that this result holds over \textit{any} interval $|I|$ of a fixed length, justifying our earlier claim that DtACI is able to adapt to the distribution shift locally over all time steps. 

Unfortunately, the values for $\sigma$ and $\eta$ specified above are not usable in practice since they depend on both the size of the time interval $|I|$ and the non-constant value $\sum_{t=r}^s \mme[\ell(\beta_t,\alpha_t)]^2$. For this first issue, the user can pick any interval size of interest,  with the consideration that choosing larger intervals gives a tighter bound at the cost of weaker local guarantees. In our experiments, we will set $\eta$ and $\sigma$ using the choice $|I| = 500$.

\sloppy
For the second issue, we give two options. The first, is to note that in the idealized setting, where there is no distribution shift, we would have $\beta_t \sim \text{Unif}(0,1)$ and $\alpha^i_t \cong \alpha^*_t = \alpha$. Plugging in these approximations we obtain 
\[
\frac{1}{|I|}\sum_{t=r}^s \mme[\ell(\beta_t,\alpha_t)^2] \cong \mme_{\beta \sim \text{Unif(0,1)}}[\ell(\beta,\alpha)^2] = \frac{(1-\alpha)^2 \alpha^2 }{3},
\]
and substituting this value into the expression for $\eta$ above gives the choice $\eta = \sqrt{\frac{3}{500}}\sqrt{\frac{\log(k\cdot 500) + 2}{(1-\alpha)^2 \alpha^2}}$. 

Our second option, is to simply update $\eta$ in an online fashion through the equation
\[
 \eta = \eta_t := \sqrt{\frac{\log(k\cdot 500) + 2}{\sum_{s=t-501}^{t} \mme[\ell(\beta_s,\alpha_s)^2]}}.
 \] 
This choice would allow us to adaptively track any changes in $\frac{1}{500} \sum_{s=t-501}^{t} \mme[\ell(\beta_s,\alpha_s)^2]$ across time. In the Appendix, we prove a generalization of Theorem \ref{thm:dynamic_reg} that allows $\eta = \eta_t$ to vary across time. We find that a dynamic choice of $\eta_t$ offers the same regret guarantees as a fixed choice so long as the variability in $\eta_t$ is not too large. Hence, adaptive values for $\eta_t$ can be used to minimize the regret bound of Theorem \ref{thm:dynamic_reg}. On the other hand, empirically, we find that on real data the approximation $\frac{1}{|I|}\sum_{t=r}^s \mme[\ell(\beta_t,\alpha_t)^2] \cong \frac{(1-\alpha)^2 \alpha^2}{3}$ is highly accurate. Thus, the two different choices for $\eta$ give nearly identical results in practice. For ease of presentation in the sections that follow, we will only display results using the first fixed, heuristic choice of $\eta$. Results for the variable choice are given in the Appendix.

\subsection{Bounds on the short-term coverage}\label{sec:coverage_bounds}

The previous section gives bounds on the performance of $\alpha_t$ in terms of the pinball loss $\ell(\beta_t,\alpha_t)$. However, the pinball loss is not our true objective and our primary goal is to obtain a value of $\alpha_t$ that is close to $\alpha^*_t$. Our next result provides a direct connection between bounds on $\ell(\beta_t,\alpha_t)$ and bounds on $(\alpha_t - \alpha^*_t)^2$.

\begin{proposition}\label{prop:loss_connection}
Let $\beta$ be a random variable and assume that there exists a value $\alpha^*$ such that $\mmp(\beta < \alpha^*) = \alpha$. Then, for any $\tau$,
\begin{align*}
& \mme[\ell(\beta,\tau)] - \mme[\ell(\beta,\alpha^*)] = \begin{cases}
\mme[(\tau - \beta) \bone_{\alpha^* < \beta \leq \tau}] \text{, if } \tau \geq \alpha^*,\\
\mme[(\beta - \tau) \bone_{\tau < \beta \leq \alpha^*}] \text{, if } \tau < \alpha^*.
\end{cases}
\end{align*}
So, in particular, if $\beta$ has a density $p(\cdot)$ on $[0,1]$ with $p(x) \geq p > 0$ for all $x \in [0,1]$, then
\[
\mme[\ell(\beta,\tau)] - \mme[\ell(\beta,\alpha^*)] \geq \frac{p(\tau - \alpha^*)^2}{2}.
\]
\end{proposition}

Now, let $\alpha^*_t$ be any value satisfying $\mmp(Y_{t} \in \hat{C}_t(\alpha^*_t) | \{\beta_s\}_{s<t}) = 1-\alpha$. Then, combining Proposition \ref{prop:loss_connection} with the results from Section \ref{sec:reg_bounds} we obtain the desired bound on $(\alpha_t - \alpha^*_t)^2$,
\begin{equation}\label{eq:final_reg_bound}
 \frac{1}{|I|} \sum_{t=r}^s \frac{p\mme[(\alpha_t - \alpha_t^*)^2]}{2} \leq O\left(\sqrt{\frac{\log(|I|)}{|I|}} \right) + O\left(\sqrt{ \frac{ \sum_{t=r+1}^s \mme[|\alpha^*_t - \alpha^*_{t-1}|]}{|I|}} \right) ,
\end{equation}
where the expectation is now over the randomness in both Algorithm \ref{alg:faci} and $\{\beta_t\}_{t \leq s}$, and $p$ is any lower bound on the density of $\beta_t, | \{\beta_r\}_{r < s}$, $\forall t \leq s$. Similarly, if we additionally assume that $\beta \mapsto \mmp(Y_{t} \in \hat{C}_t(\beta) | \{\beta_s\}_{s<t})$ is $L$-Lipschitz, then $|\mmp(Y_{t} \in \hat{C}_t(\beta) | \{\beta_s\}_{s<t}) -(1-\alpha)| \leq L|\alpha_t-\alpha^*_t|$ and thus (\ref{eq:final_reg_bound}) can also be read as a bound on the local coverage of DtACI. Such a Lipschitz assumption may be reasonable if the distribution of the conformity scores and our estimates of its quantiles are sufficiently smooth.

In \cite{Gibbs2021}, the authors showed that adaptive conformal inference satisfies a similar bound to (\ref{eq:final_reg_bound}). However, their result requires three major assumptions: 1) the size of the distribution shift $\frac{ \sum_{t=r+1}^s |\alpha^*_t - \alpha^*_{t-1}|}{|I|}$ is known, 2) $(X_t,Y_t)$ is generated by a hidden Markov model, and 3) the regression model is not re-fit across time. In stark contrast, our result makes no such assumptions and, in addition, is adaptive to changes in the size of the distribution shift across time. In practical settings, none of these three assumptions can be reasonably expected to hold and thus our results constitute a significant generalization of those in \cite{Gibbs2021}.

\subsection{Bounds on the long-term coverage}\label{sec:long_term_cov}

The results of the previous section show that DtACI obtains a coverage rate close to $1-\alpha$ over any local time interval. It is natural to ask if elongating this interval leads to an average coverage of exactly $1-\alpha$. Here, we show that this is indeed the case if the parameters $\eta = \eta_t \to 0$ and $\sigma = \sigma_t \to 0$. At a high-level, sending $\eta$ and $\sigma$ to 0 causes the method to put more weight on older historical data and thus gives a version of DtACI that is closer to AgACI and MVP. Prior work has shown that MVP also obtains exact long-term coverage (\cite{Bastani2022}), while for AgACI, this property has only been observed empirically (\cite{Zaffran2022}).

\begin{theorem}\label{thm:long_term_coverage}
Consider a modified version of Algorithm 1 in which on iteration $t$ the parameters $\eta$ and $\sigma$ are replaced by values $\eta_t$ and $\sigma_t$. Let $\gamma_{\textnormal{min}} := \min_i \gamma_i$ and $\gamma_{\textnormal{max}} := \max_i \gamma_i$. Then,
\[
\left| \frac{1}{T} \sum_{t=1}^T\mme[\textup{err}_t] - \alpha \right| \leq \frac{1+2\gamma_{\textnormal{max}}}{T\gamma_{\textnormal{min}}} +  \frac{(1+2\gamma_{\textnormal{max}})^2}{\gamma_{\textnormal{min}}}  \frac{1}{T} \sum_{t=1}^T \eta_te^{\eta_t(1+2\gamma_{\textnormal{max}})}+ 2  \frac{1+\gamma_{\textnormal{max}}}{\gamma_{\textnormal{min}}}\frac{1}{T} \sum_{t=1}^T \sigma_t,
\] 
where the expectation is over the randomness in Algorithm \ref{alg:faci} and the data $\beta_1,\dots,\beta_T$ can be viewed as fixed. So, in particular, if $\lim_{t \to \infty} \eta_t = \lim_{t \to \infty} \sigma_t = 0$, then $\lim_{T \to \infty} \frac{1}{T} \sum_{t=1}^T\textup{err}_t \stackrel{a.s.}{=} \alpha$.
\end{theorem}

Following the discussion of the previous sections, decaying values of $\eta_t$ and $\sigma_t$ should only be used if the size of the distribution shift in the environment is known to be stationary. Since we do not consider this to be a realistic assumption in most situations, we advocate for using constant or slowly varying values for $\eta$ and $\sigma$. For these choices, we empirically find that DtACI can produce intervals that are biased in the sense that $\lim_{T \to \infty} \frac{1}{T} \sum_{t=1}^T \text{err}_t \neq \alpha$. However, in all the examples we have investigated, this bias is sufficiently small to be of little practical consequence.

\subsection{Removing randomness in the choice of $\pmb{\alpha_t}$}

In practical settings, the randomness in the choice of $\alpha_t $ may be undesirable. To rectify this, we provide an alternative approach that replaces the choice $\alpha_t \sim \sum_{i =1}^k p_t^i\delta_{\alpha_t^i}$ with $\bar{\alpha}_t =  \sum_{i =1}^k p_t^i\alpha_t^i$.  The full version of this procedure is stated in Algorithm \ref{alg:faci_avg}. Importantly, this new method admits the same regret bound as our original procedure.  

\begin{corollary}
Under the same conditions, the conclusion of Theorem \ref{thm:dynamic_reg} holds with $\alpha_t$ replaced by $\bar{\alpha_t}$ on the left-hand side of the display.
\end{corollary}
\begin{proof}
This is an immediate consequence of Jensen's inequality.
\end{proof}

\begin{algorithm}
 \KwData{Observed values $\{\beta_t\}_{1 \leq t \leq T}$, set of candidate $\gamma$ values $\{\gamma_i\}_{1 \leq i \leq k}$, starting points $\{\alpha_1^i\}_{1 \leq i \leq k}$, and parameters $\sigma$ and $\eta$.}
 $w^i_1 \leftarrow 1,\ 1 \leq i \leq \in k$\;
 \For{$t=1,2,\dots,T$}{
 	Define the probabilities $p_t^i := w^i_t/\sum_{1 \leq j \leq k}w^{j}_t$, $\forall 1 \leq i \leq k$\;
 	Output $\bar{\alpha}_t =  \sum_{1 \leq i \leq k} p_t^i\alpha^i_t$\;
 	$\bar{w}^i_t \leftarrow w^i_t \exp(-\eta \ell(\beta_t,\alpha_t^i)),\ \forall1 \leq i \leq k$\;
 	$\bar{W}_t \leftarrow \sum_{1 \leq i \leq k} \bar{w}^i_t $\;
 	$w^i_{t+1} \leftarrow (1-\sigma) \bar{w}^i_t + \bar{W}_t \sigma/k$\;
 	$\text{err}^i_{t} := \bone\{Y_t \notin \hat{C}_t(\alpha^i_t)\}$, $\forall 1 \leq i \leq k$\;
 	$\text{err}_{t} := \bone\{Y_t \notin \hat{C}_t(\bar{\alpha}_t)\}$\;
 	$\alpha^i_{t+1} = \alpha^i_{t} + \gamma_i(\alpha - \text{err}^i_{t})$, $\forall 1 \leq i \leq k$\;
 }
 \caption{}
 \label{alg:faci_avg}
\end{algorithm}

Unsurprisingly, we find that in practice Algorithms \ref{alg:faci} and \ref{alg:faci_avg} produce nearly identical results. Thus, we will abuse terminology and also refer to Algorithm \ref{alg:faci_avg} as DtACI. For simplicity, the following sections show results only for Algorithm \ref{alg:faci_avg}.

\section{Empirical results}

In this section we investigate the performance of DtACI as well as the previously proposed AgACI and MVP methods. In all experiments the set of candidate $\gamma$ values is taken to be $\{$0.001, 0.002, 0.004, 0.008, 0.0160, 0.032, 0.064, 0.128$\}$. Code for reproducing  these results can be found at \url{https://github.com/isgibbs/DtACI}.

\subsection{Simulated examples}\label{sec:simulated_data}

We begin by considering a set of simulated examples in which we can exactly measure the local coverage properties of the methods. To make the results interpretable, we focus on a simple setting in which we observe a sequence of independent random variables $\{Y_t\}_{t=1}^T$, where $Y_t \sim \mathcal{N}(\mu_t,1)$, and form prediction sets using the standard normal distribution as our quantile estimator, i.e. we set
\[
\hat{C}_t(\alpha_t) := \left\{ y : y \leq \text{Quantile}(1-\alpha_t,\mathcal{N}(0,1)) \right\}.
\]
We consider three different choices for the sequence of means $\{\mu_t\}_{t=1}^T$:
\begin{itemize}
    \item A \textit{stationary} setting in which $\mu_t = 0$ is held constant and thus the data are i.i.d..
    \item A \textit{smooth} shift setting in which $\mu_t$ drifts continuously across time. More precisely we set $\mu_0=0$ and 
    \[
    \mu_{t+1} = \mu_t + \frac{1}{2}(\mu_t - \mu_{t-1}) + \frac{1}{2}\epsilon_t
    \]
    where $\{\epsilon_t\} \stackrel{i.i.d.}{\sim} \mathcal{N}(0,0.006)$.
    \item A setting with \textit{jump} shifts in which $\mu_t$ undergoes jump discontinuities of various sizes. In particular, we divide the time period into three equally sized intervals. In the first and third interval $\mu_t$ oscillates between -0.075 and 0.075, while in the second interval the distribution shifts are larger and $\mu_t$ oscillates between $-1.5$ and $1.5$. 
\end{itemize}

\begin{figure}[ht]
    \centering
    \includegraphics[scale=0.24]{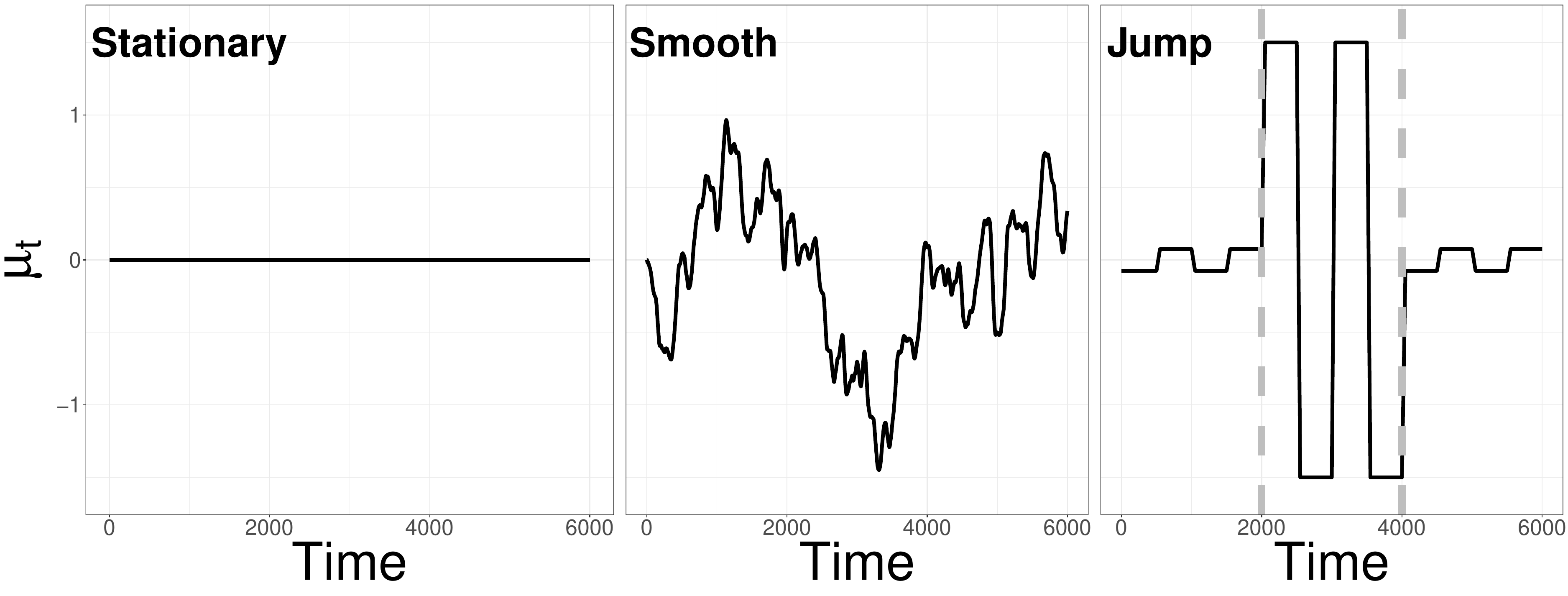}
        \caption{Trajectories for $\mu_t$ in the \textit{stationary}, \textit{smooth}, and \textit{jump} settings. To aid readability, the trajectory of $\mu_t$ in the center panel has been locally averaged over a moving time interval of width $50$. Vertical grey lines in the jump shift plot denote the regime switches where the size of the distribution shift changes.}

    \label{fig:simulated_mus_forreal}
\end{figure}

\begin{figure}[h]
    \centering
    \includegraphics[scale=0.24]{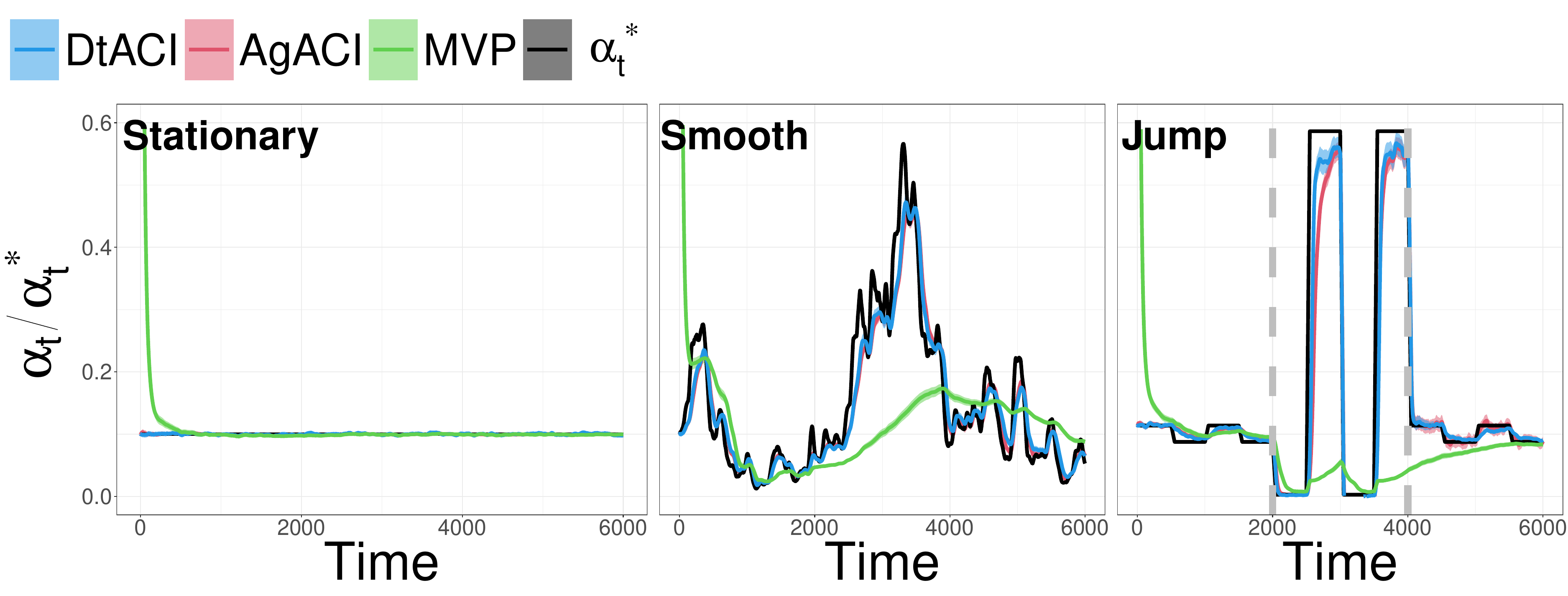}
        \caption{Comparison of the simulated trajectories of $\alpha_t^*$ (black) against the estimated values, $\alpha_t$ output by DtACI (blue), AgACI (red), and MVP (green). Solid lines display averages across 100 trials with $\mu_t$ (and thus $\alpha_t^*$) held fixed and $\alpha_t$ regenerated using independent draws of $\{Y_t\}$. Shaded regions show confidence intervals for the corresponding means. To aid readability, means and confidence intervals have been locally averaged over a moving time interval of width $50$. Finally, the vertical grey lines in the jump plot denote the regime switches where the size of the distribution shift changes.}

    \label{fig:simulated_mus}
\end{figure}

The final trajectories of $\mu_t$ generated in all three settings are shown in Figure \ref{fig:simulated_mus_forreal}. Figure \ref{fig:simulated_mus} shows the corresponding trajectories of $\alpha_t^*$ as well as the estimates, $\alpha_t$, produced by DtACI, AgACI, and MVP. We can immediately see that DtACI and AgACI accurately adapt to the local changes in $\alpha_t^*$ across all three settings, while MVP only performs well in the stationary environment.

\begin{figure}
    \centering
    \includegraphics[scale=0.18]{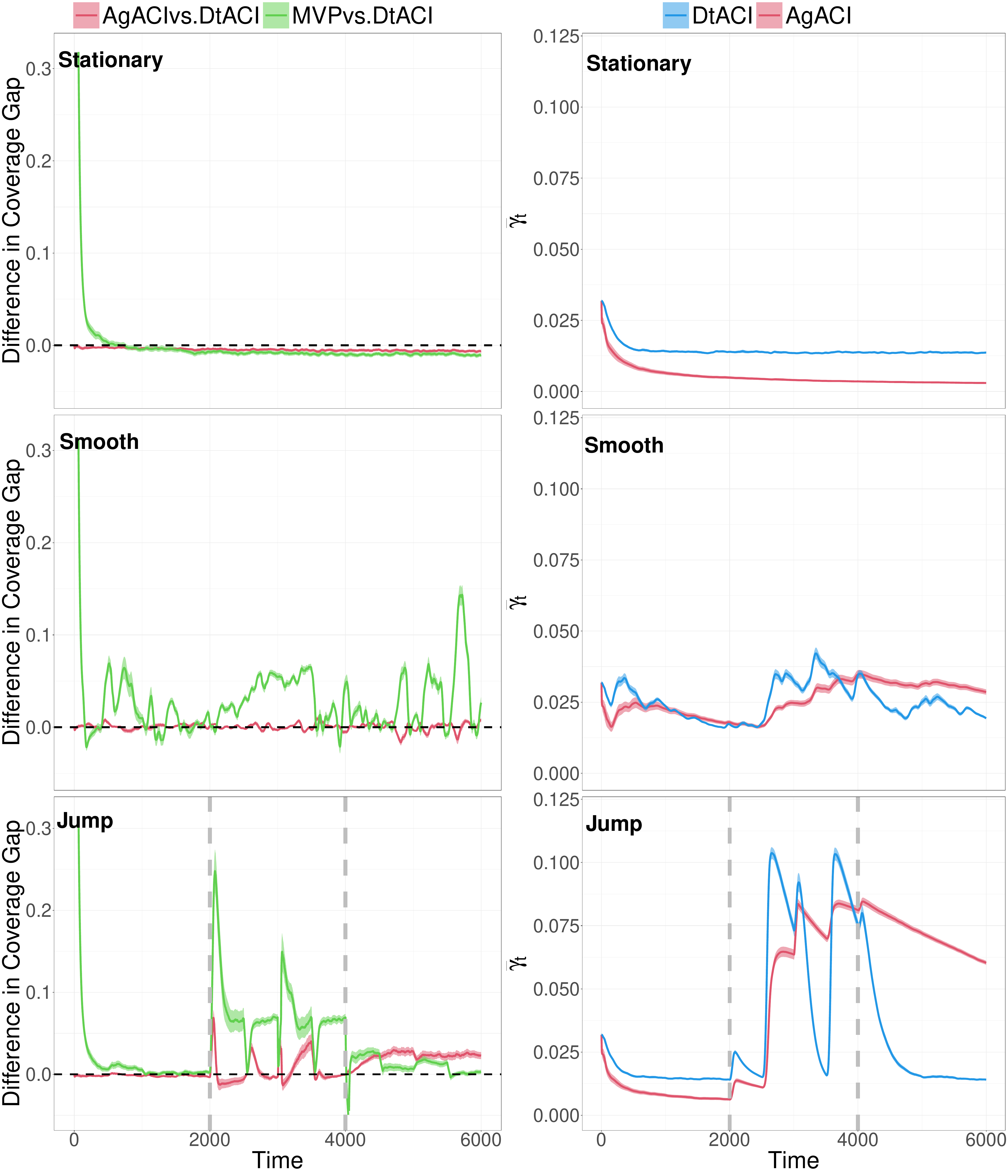}
    \caption{Comparison of the coverage gaps obtained by AgACI and MVP against the baseline of DtACI (left-panels) and the mean step-size trajectories output by DtACI and AgACI (right-panels). Solid lines display averages across 100 trials, while shaded regions show confidence intervals for the corresponding means. To aid readability, means and confidence intervals have been locally averaged over a moving time interval of width $50$. Finally, vertical grey lines in the jump shift plots denote the regime switches where the size of the distribution shift changes, while horizontal black lines in the coverage gap plots denote the value $0$ at which the performance of DtACI exactly matches the competing methods.}
    
    \label{fig:simulated_res}
\end{figure}

To get a more fine-grained comparison of the relative performance of DtACI, AgACI, and MVP we additionally compute the time-instantaneous coverages of these methods. Namely, let $\alpha_t^D$, $\alpha_t^A$, and $\alpha_t^M$ denote the values output by DtACI, AgACI, and MVP at time step $t$. Then, we compute the instantaneous coverage gaps $\text{CG}_x = |\mmp(Y_t \in \hat{C}(\alpha_t^x)|\alpha_t^x) - (1-\alpha)|$ for $x \in \{D,A,M\}$ and we plot the difference between the values obtained by AgACI and MVP and those obtained by DtACI (i.e. the values $\text{CG}_A - \text{CG}_D$ and $\text{CG}_M - \text{CG}_D$). Thus, in the results that follow a value of $0$ indicates identical performance, while positive/negative values indicate that DtACI performs better/worse than the competitors.

The resulting coverage performances are shown in Figure \ref{fig:simulated_res}. Overall, we find that DtACI offers greater adaptivity and more precise coverage than AgACI and MVP in the non-stationary settings, while suffering only a slight degradation in performance under stationarity. More specifically, our results for each of the three settings can be summarized as follows.
\begin{itemize}
    \item \textit{Stationary setting:} In the stationary setting MVP and AgACI slightly outperform DtACI. This is due to the fact that MVP and AgACI are able to more precisely converge to the single optimal value for $\alpha_t^*$, while DtACI can never set its step-size to exactly 0 and thus maintains some minor fluctionations around $\alpha_t^*$ across all time steps.
    \item \textit{Smooth shift setting:} For non-stationary data MVP now gives significantly worse performance and shows little adaptivity to the distribution shifts. On the other hand, both AgACI and DtACI perform well and are able to approximately track the changes in $\alpha_t^*$ over time. Overall, the results for AgACI and DtACI are nearly identical, which is expected since a single choice of the step-size is sufficient to give good performance in this environment.
    \item \textit{Sharp shift setting:} Once again MVP fails to adapt to the distribution shifts. Additionally, while AgACI does show reasonable adaptivity, it fails to adjust its step-size to track the changes in the environment. We visualize this behaviour in the right panels of Figure \ref{fig:simulated_res}, which show the average step-sizes, $\bar{\gamma}_t = \sum_{i=1}^k p_t^i \gamma_i$ produced by AgACI and DtACI. We see that AgACI correctly gives a small step-size in the first phase when the distribution shifts are small. However, once when we enter the second stage and the distribution shifts jump in magnitude, AgACI is slow to react and its step-size lags behind that of DtACI. This behaviour is amplified in the third stage during which AgACI never decreases its step-size to match the smaller distribution shifts and thus suffers large coverage errors throughout.  
\end{itemize}
Overall, we find that while AgACI and MVP perform slightly better than DtACI in the stationary setting, this comes at the cost of a greater loss of adaptivity in the non-stationary settings.

\subsection{Real data examples}\label{sec:real_data}

\subsubsection{Online prediction in the stock market}\label{sec:stocks}

For our first real data example, we return to a stock market prediction task that was originally used to evaluate ACI (\cite{Gibbs2021}). In this problem, the goal is to use the previously observed values of a stock price, $\{P_s\}_{s=0,\dots,t-1}$, to predict its volatilty at the next step, defined as
\[
V_t :=  \left(\frac{P_t - P_{t-1}}{P_{t-1}}\right)^2.
\]
To predict the volatility, we model the stock returns $R_t := \frac{P_t - P_{t-1}}{P_{t-1}}$ as coming from a GARCH(1,1) design. This is a classical financial model for market dynamics in which it is assumed that $R_t = \sigma_t\epsilon_t$ with $\epsilon_t \sim \mathcal{N}(0,1)$ and 
\[
\sigma^2_t = \omega + \tau V_{t-1} + \lambda \sigma_{t-1}^2,
\]
for some unknown parameters $\omega$, $\tau$, $\lambda$. At each time step, $t$, we use the most recent 1250 days of returns $\{R_s\}_{t-1250 \leq s<t}$ to produce estimates $\hat{\omega}_t $, $\hat{\tau}_t$, $\hat{\lambda}_t$, $\{\hat{\sigma}^s_t\}_{s<t}$, and a one-step ahead prediction $(\hat{\sigma}^t_t)^2 = \hat{\omega}_t + \hat{\tau}_t V_{t-1} + \hat{\lambda}_t (\hat{\sigma}^{t-1}_{t})^2$. We then construct prediction sets using the equation
\[
\hat{C}_t(\alpha_t) := \left\{v : S_t(v) \leq \text{Quantile}\left(1-\alpha_t, \frac{1}{1250} \sum_{s=t-1250}^{t-1} \delta_{S_s(V_s)} \right) \right\},
\]
where $S_t(v)$ is taken to be either the normalized conformity score $S_t(v) = |v - (\hat{\sigma}^t_t)^2|/(\hat{\sigma}^t_t)^2$ or the unnormalized score $S_t(v) := |v -(\hat{\sigma}^t_t)^2|$. 

\begin{figure}[h]
\begin{centering}
\includegraphics[scale=0.2]{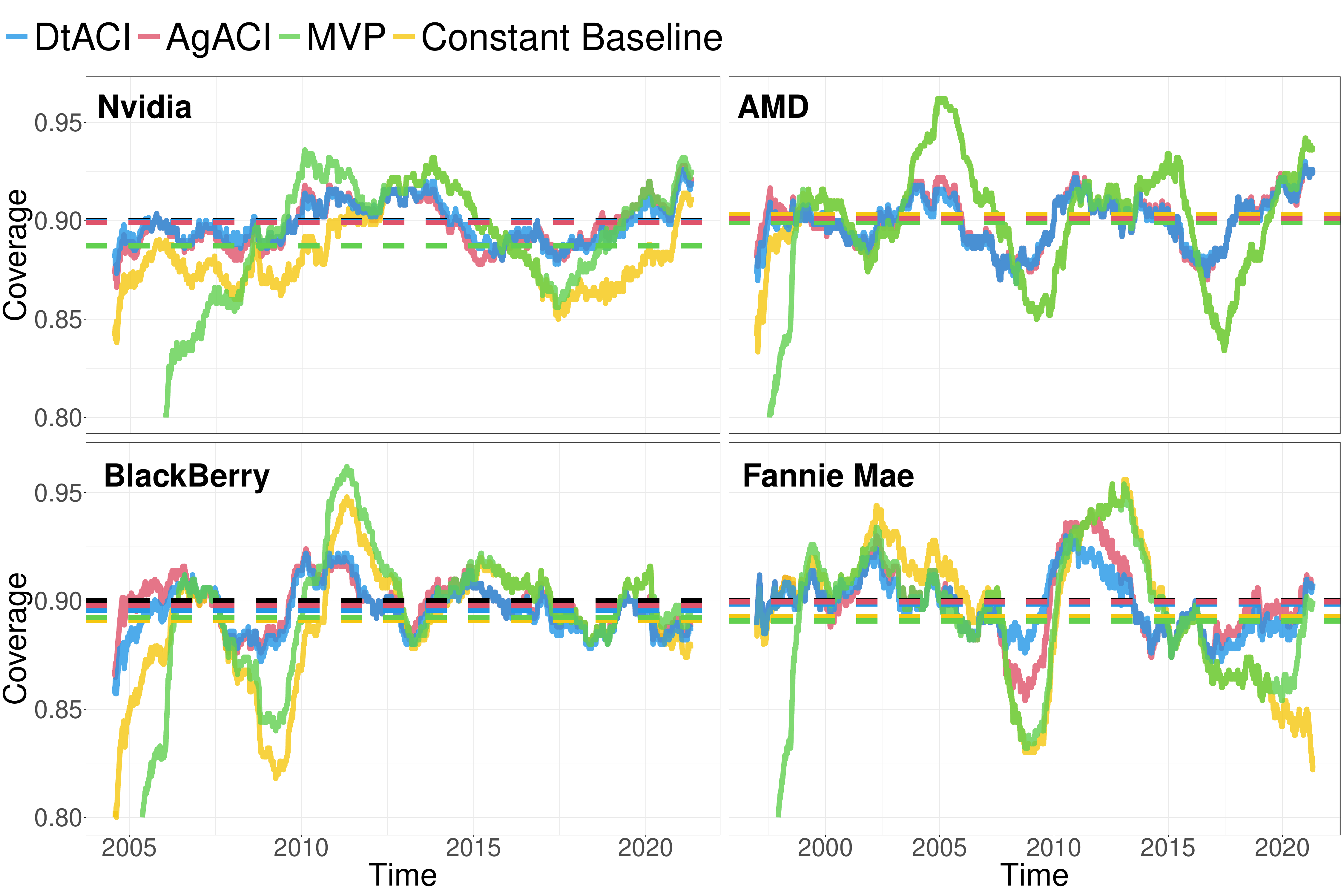}
\caption{Estimation of stock market volatility using conformity score $S_t(v) = |v - (\hat{\sigma}^t_t)^2|/(\hat{\sigma}^t_t)^2$. Solid lines show the local coverage level for DtACI (blue), AgACI (red), MVP (green), and a naive baseline that holds $\alpha_t = \alpha$ fixed (yellow). Dashed lines indicate the global coverage frequency over all time steps for the same methods. Finally, the black dashed lines indicates the target level of $1-\alpha = 0.9$. Note that in some of the panels (e.g.~the top-right AMD panel) the yellow and green lines exactly overlap for many time steps leaving only the green line clearly visible. Some overlap is also observed in the red and blue lines.}
\label{fig:stock_normalized}
\end{centering}
\end{figure} 

In \cite{Gibbs2021} it was found that the normalized and unnormalized conformity scores lead to distribution shifts of drastically different sizes. To gain some intuition as to why this is the case, observe that if the GARCH(1,1) model is true, and moreover $(\hat{\sigma}^t_t)^2 = \sigma^2_t$ is an exactly accurate prediction, then the normalized score is distributed as $S_t(V_t) \sim \chi^2_1$, while the unnormalized score follows $S_t(V_t) \sim \sigma_t^2\chi^2_1$. Thus, in this setting, $\alpha^*_t$ will be invariant across time for the normalized score and highly variable for the unnormalized score. Consistent with this intuition, \cite{Gibbs2021} found that the distribution shift is much larger for the unnormalized score than the normalized score and, thus, in order to obtain good local coverage, ACI requires different (and \textit{a priori} unknown) step sizes for the two scores. In contrast, as we will show shortly, DtACI obtains good performance in both conditions without any prior knowledge of the distribution shift.

Similar to the previous section, we measure the performance of the prediction sets by their local coverage. Since the true time instantaneous coverage is no longer an observable quantity, we instead compute empirical local average coverage rates over a moving 500-day window, i.e. we compute the moving average
\begin{equation}\label{eq:local_cov}
\text{LocalCov}_t :=1- \frac{1}{500} \sum_{t - 250 + 1}^{t+250} \text{err}_t.
\end{equation}

\begin{figure}[ht]
\begin{centering}
\includegraphics[scale=0.2]{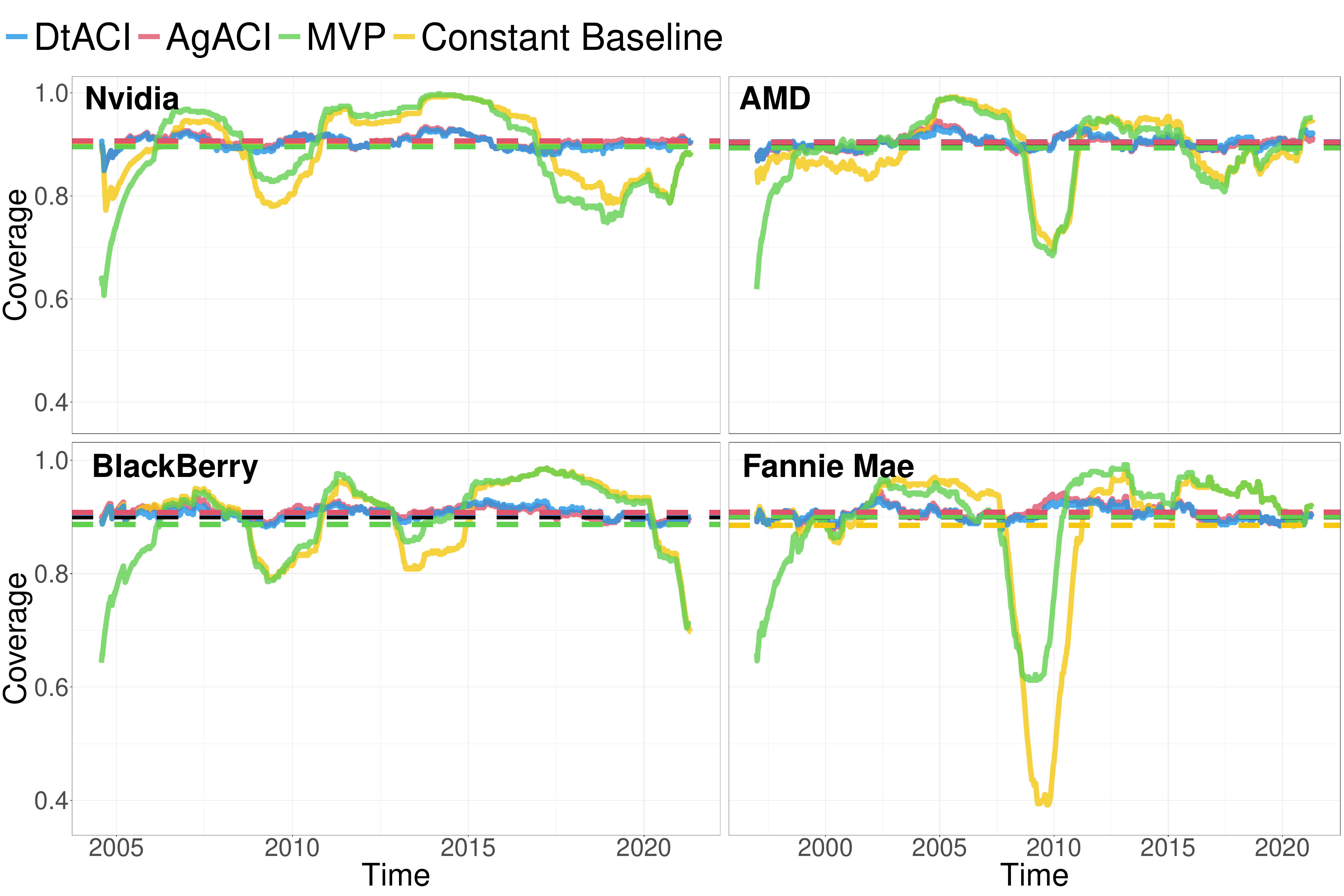}
\caption{Estimation of stock market volatility using conformity score $S_t(v) = |v - (\hat{\sigma}^t_t)^2|$. Solid lines show the local coverage level for DtACI (blue), AgACI (red), MVP (green), and a naive baseline that holds $\alpha_t = \alpha$ fixed (yellow). Dashed lines indicate the global coverage frequency over all time steps for the same methods. Finally, the black dashed lines indicates the target level of $1-\alpha = 0.9$.}
\label{fig:stock_unnormalized}
\end{centering}
\end{figure} 

Figures \ref{fig:stock_normalized} and \ref{fig:stock_unnormalized} show the local average coverage rates obtained by DtACI, AgACI, and MVP, using the normalized and unnormalized conformity scores on four different stocks. In addition, we have also plotted the coverage values for a naive baseline method that holds $\alpha_t = \alpha$ fixed. At a high level, this baseline essentially measures the size of the underlying shifts in the environment. Finally, as an additional reference, Figure \ref{fig:stock_prices} in the Appendix shows the prices of the stocks over the same period. 

All four stocks demonstrate obvious price swings leading to clear distribution shifts in the data. Similar to our simulated examples, we find that DtACI is able adapt to both the presence and size of these shifts and obtain a local coverage rate near the target level of $1-\alpha = 0.9$ over all time steps and conditions. Overall, the size of the distribution shifts in these environments appears relatively constant, and thus AgACI produces nearly identical results to DtACI. Perhaps the only exception to this is the Fannie Mae data in Figure \ref{fig:stock_normalized} for which there is a small time window near the 2009 financial crisis where AgACI is slow to react to the sharp rise in volatility. Finally, we find that MVP shows minimal adaptivity to the underlying shifts throughout and often performs nearly identically to the fixed baseline method.

While DtACI seems to perform reasonably well across all settings, one may still wonder if additional improvements can be made. Namely, is it possible for a sensible method to produce local coverage errors that are tighter to the $1-\alpha$ line? To answer this question, we compare the coverage properties of DtACI against the ideal situation in which the coverage errors are an  i.i.d.~Bernoulli($\alpha$) sequence. Namely, Figure \ref{fig:bern_comp} shows Q-Q plots comparing the empirical distributions of $|\text{LocalCov}_t - (1-\alpha)|$ realized by DtACI against the distribution of the same quantity for an i.i.d. Bernoulli($\alpha$) sequence. We find that the two distributions nearly exactly align across all four stocks. Thus, in this experiment the local coverage properties of DtACI are in some sense unimprovable. As a final remark, note that while Figure \ref{fig:bern_comp} only shows results for the normalized conformity scores, we also obtained nearly identical results for the unnormalized scores (see Figure \ref{fig:bern_comp_bad} in the Appendix).

\begin{figure}[h]
\begin{centering}
\includegraphics[scale=0.2]{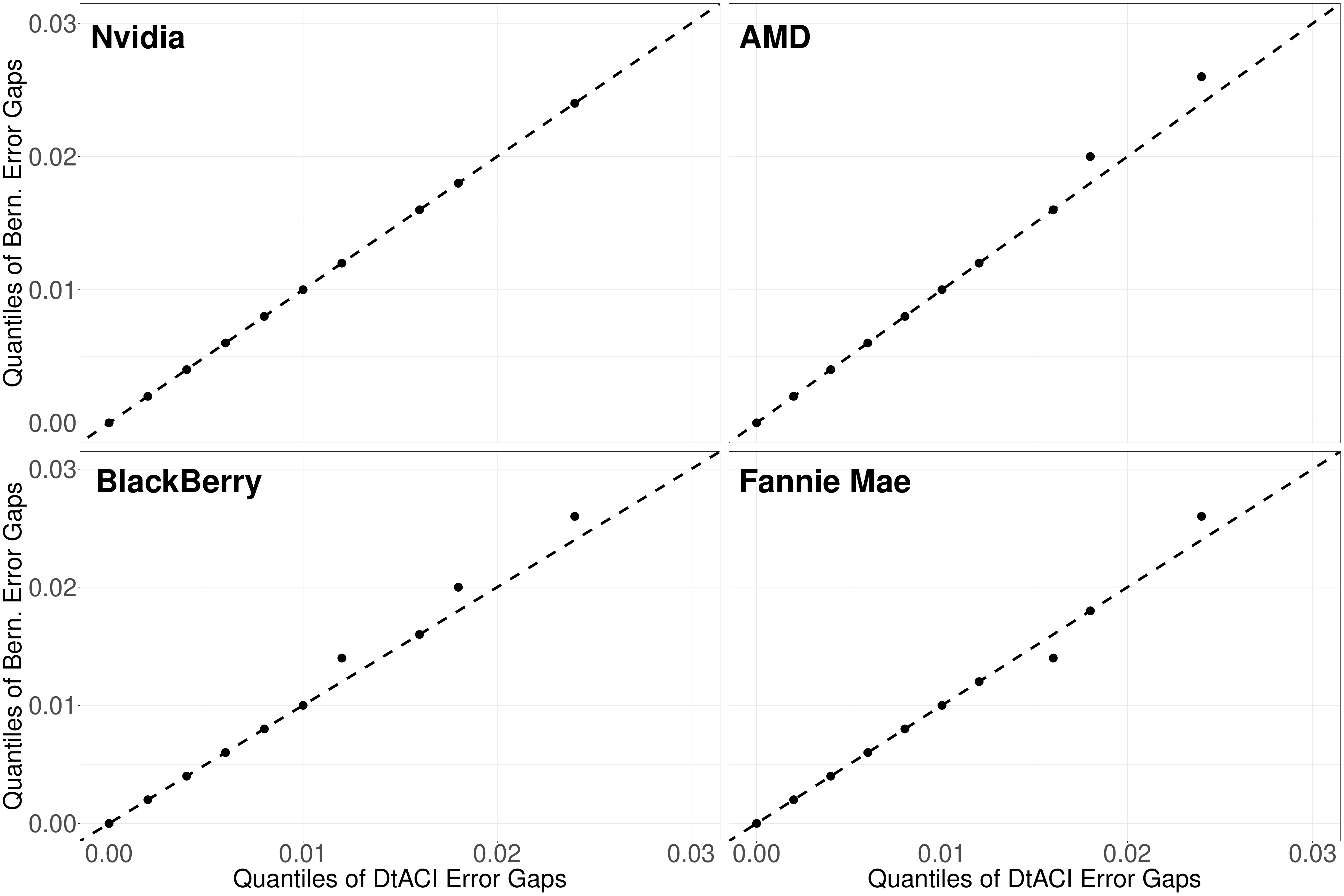}
\caption{Q-Q plots comparing the distribution of local coverage gaps obtained by an i.i.d. Bernoulli($\alpha$) sequence against those realized by DtACI. The dashed line indicates the ideal situation of exact equality.}
\label{fig:bern_comp}
\end{centering}
\end{figure}

\subsubsection{Evaluating the reactivity of DtACI}\label{sec:cond_on_alphat_cov}

In \cite{Bastani2022}, adaptive conformal inference was criticized for failing to provide a coverage guarantee conditional on the value of $\alpha_t$. Indeed, one may be concerned that since ACI acts reactively by widening/shrinking its prediction sets in response to past mistakes, good local coverage is obtained not due to successfully learning $\alpha^*_t$, but rather as a result of the simple tendency of the algorithm to correct its prior under/over-coverage. Our simulations in Section \ref{sec:simulated_data} already partially show that DtACI can successfully learn $\alpha^*_t$. Here, we provide further evidence demonstrating that DtACI is not simply acting reactively on real data.

To do this, we divide the interval $[0,1]$ into $m$ evenly sized sub-intervals $B_1,\dots,B_m$ and evaluate the conditional coverage levels 
\begin{equation}\label{eq:cond_coverage}
\text{CondCoverage}_i := \frac{1}{ |\{t : \bar{\alpha}_t \in B_i\}|  } \sum_{t : \bar{\alpha}_t \in B_i} \text{err}_t.
\end{equation}
We apply this to the volatility dataset outlined in the previous section using the normalized conformity score and display our results in Figure \ref{fig:cond_alphat_cov}. As a visual aid, this figure contains error bars indicating the 0.05 and 0.95 quantiles of $\text{CondCoverage}_i $ across block bootstrap re-samples of the data (see Algorithm \ref{alg:block_boot} in the Appendix for details). These error bars would be expected to give a valid confidence interval for the conditional coverage if, for instance, $\{(X_t,Y_t)\}$ was a stationary time-series. However, since there is distribution shift in these examples, the error bars are \textit{not} accompanied by any coverage guarantee. Thus, we present them simply as a visual aid to help the reader judge the distance between $\text{CondCoverage}_i$ and $1-\alpha$ relative to the sample size $|\{t : \bar{\alpha}_t \in B_i\}|$. 

\begin{figure}[ht]
\begin{centering}
\includegraphics[scale=0.18]{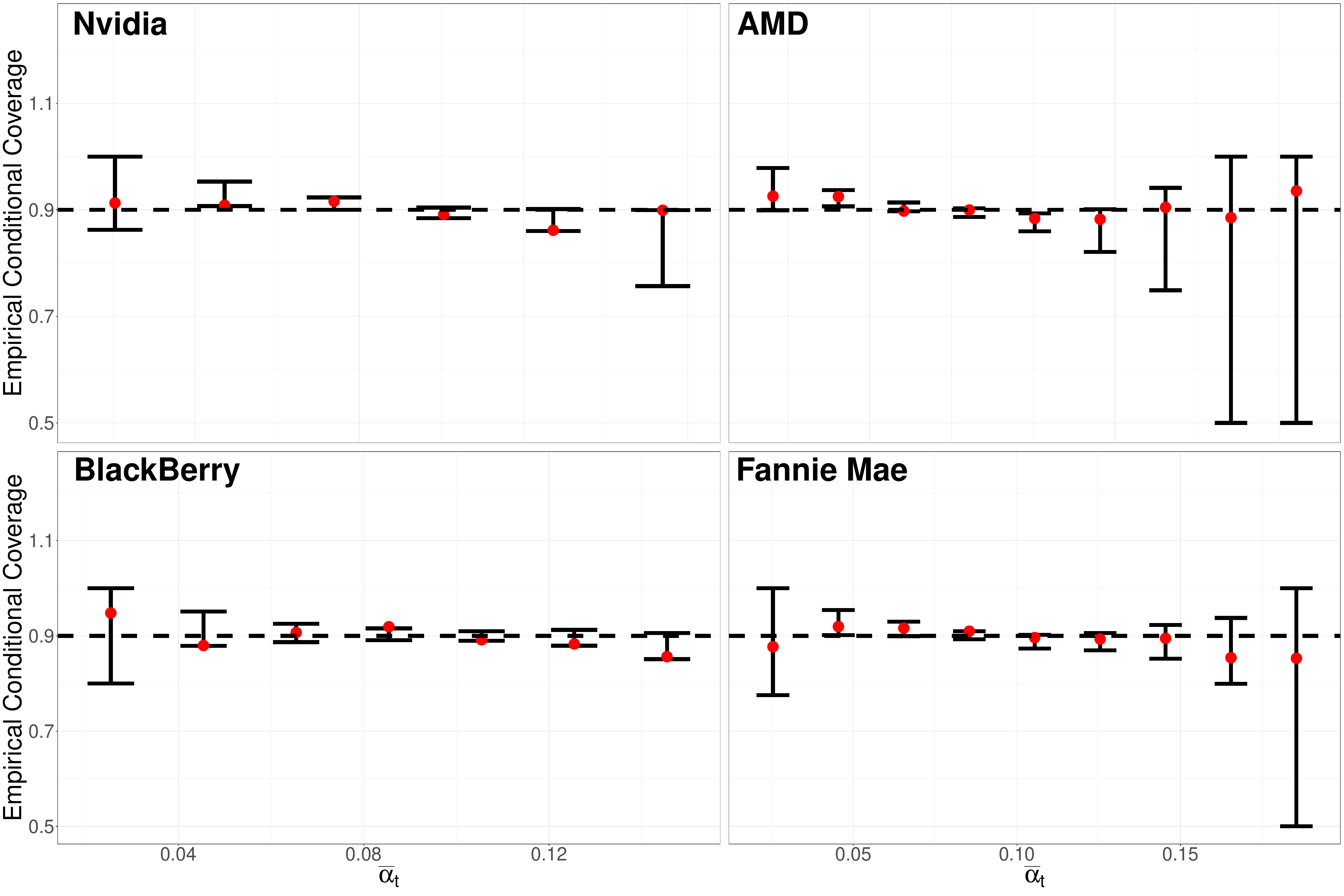}
\caption{Empirical conditional coverage of $\bar{\alpha}_t$ for the estimation of stock market volatility with conformity score $S_t(v) = |v - (\hat{\sigma}^t_t)^2|/(\hat{\sigma}^t_t)^2$. Red points show the empirical conditional coverage given $\bar{\alpha}_t \in B_i$ with error bars indicating the corresponding 0.05 and 0.95 quantiles across 100 block bootstrap resamples of the data $\{(X_i,Y_i)\}$ with block-size 100. For visual clarity all error bars are truncated to the range $[0.5,1]$. Black dashed lines shows the target level of $1-\alpha = 0.9$. }
\label{fig:cond_alphat_cov}
\end{centering}
\end{figure} 

Overall, we find that almost all of the error bars cover the target level of $1-\alpha = 0.9$ and for a large majority of the bins, $\text{CondCoverage}_i$ is nearly exactly equal to 0.9. If DtACI was simply acting reactively to its past mistakes we would expect to observe over-coverage at small values of $\alpha_t$ and under-coverage at large values of $\alpha_t$. Thus, Figure \ref{fig:cond_alphat_cov} provides strong evidence that DtACI is not enacting any such pathological behaviour. Similar results were also obtained using the unnormalized conformity scores (see Figure \ref{fig:cond_alphat_cov_unnormalized} in the Appendix). 

\subsection{Predicting Covid-19 case counts}\label{sec:covid}

Our final example considers the problem of predicting future COVID-19 case counts. We base our methods on the work of \citet{Tibs2020} and work with a simple model for generating one-week ahead forecasts of the seven day moving average of the number of confirmed cases of COVID-19 in each county in the United States. In this model, future forecasts are generated based off of the historical prevalence of COVID-19 across the US and Facebook survey data that provides us with a moving seven day average of the proportion of people who report knowing someone in their local community with COVID-19.  All data is obtained from a public repository made available by the DELPHI group at Carnegie Mellon (\cite{Reinhart2021}). 

Let $\{CO_{t,i}\}_{t,i}$ and $\{F_{t,i}\}_{t,i}$ denote the time series of COVID-19 case counts and Facebook survey responses, respectively, where $t$ indexes time and $i$ indexes one of the 3243 counties in the US.  At each time step $t$ we predict $\{CO_{t+7,i}\}_i$ by using least-squares regression to fit the model
\[
CO_{s,i} \sim \beta^t_0 + \sum_{j=1}^3 \lambda_j^t CO_{s-7j,i} + \sum_{j=1}^3 \kappa_j^t F_{s-7j,i},\ s = t-14\dots,t,\ i =1 ,\dots, 3243,
\]
and setting
\[
\widehat{CO}_{t+7,i} = \hat{\beta}^t_0 + \sum_{j=1}^3 \hat{\lambda}_j^t CO_{t-7j,i} + \sum_{j=1}^3 \hat{\kappa}_j^t F_{t-7j,i}, \ i =1 ,\dots, 3243.
\]
Because Facebook survey data is not available for all counties at all time steps, we restrict our analysis to those counties with no missing values in the above expressions. 

\begin{figure}[ht]
\begin{centering}
\includegraphics[scale=0.2]{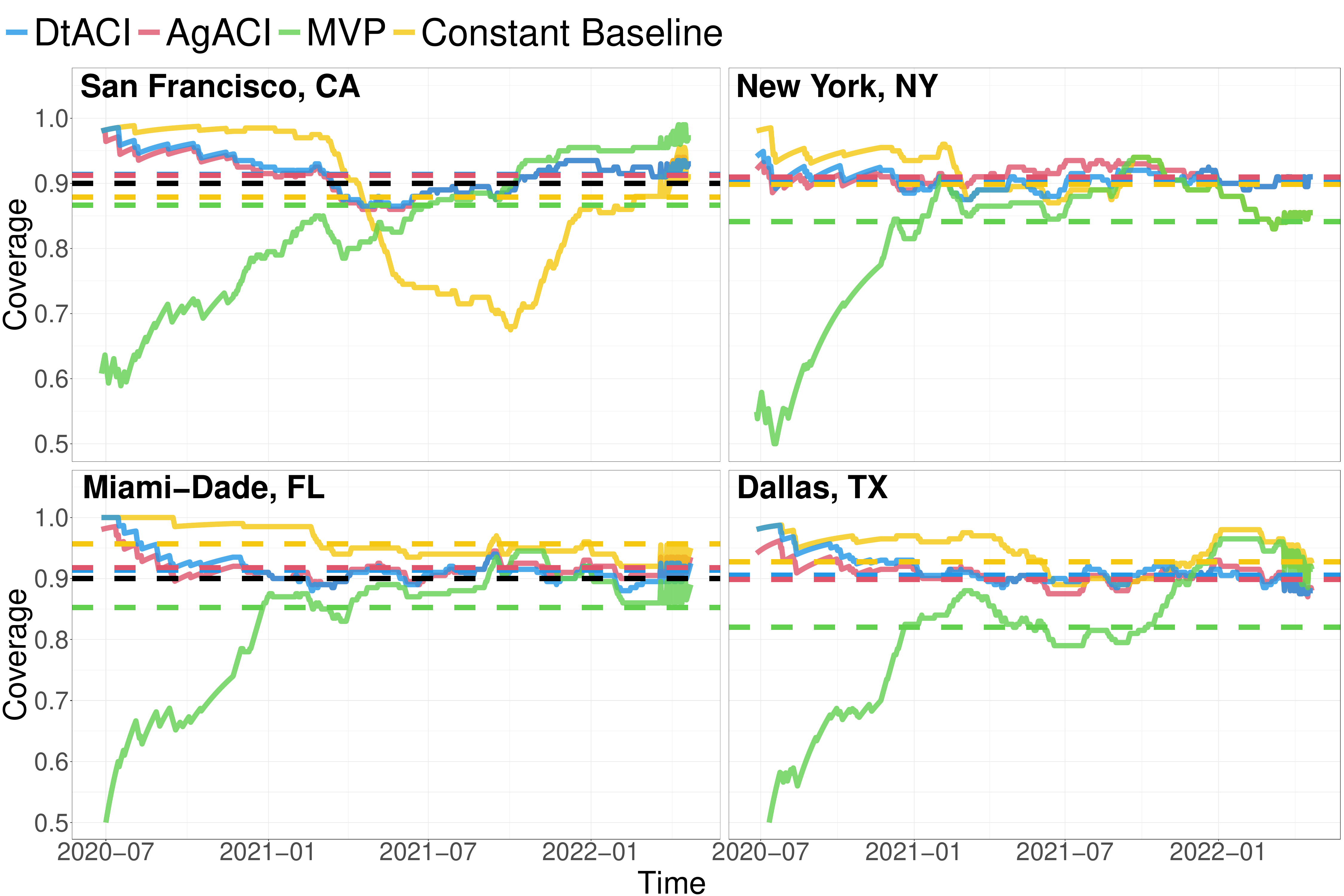}
\caption{Coverage results for the prediction of county-level COVID-19 case counts. Solid lines show the local coverage level for DtACI (blue), AgACI (red), MVP (green), and a naive baseline that holds $\alpha_t = \alpha$ fixed (yellow). Dashed lines indicate the global coverage frequency over all time steps for the same methods. The black dashed lines indicates the target level of $1-\alpha = 0.9$.}
\label{fig:covid_cov}
\end{centering}
\end{figure} 

To compute prediction sets for county $i$, we define the conformity scores $S_{t,i} := |\widehat{CO}_{t,i} - CO_{t,i}|/|CO_{t-7,i} - CO_{t,i}|$ and counts $n_t := |\{i : \text{County $i$ has available data at time } t-1\}|$, and set
\[
\hat{C}_{t,i}(\alpha_{t,i}) := \left\{ c : \frac{|\widehat{CO}_{t,i} - c|}{|CO_{t-7,i} - c|} \leq \text{Quantile}\left( 1-\alpha_{t,i}, \frac{1}{n_t} \sum_i \delta_{S_{t-1,i}}  \right) \right\}.
\]
Since different counties will undergo different dynamics at different stages of the pandemic, we run DtACI, AgACI, and MVP separately for each county to obtain a set of county-specific trajectories, $\{\alpha_{t,i}\}_{t,i}$ for $\alpha_t$.

Figure \ref{fig:covid_cov} shows the empirical local coverage rates over the nearest 200 time steps (i.e~equation (\ref{eq:local_cov}), with 250 replaced by 100) for four US counties. These four counties were chosen because they are large urban centres for which data was available over the entire time window we considered.  All four of these counties have undergone multiple waves of COVID-19, each of which caused a large swing in the observed case count (see Figure \ref{fig:covid_case_counts}) and thus induced a clear distribution shift into the data. Much like the previous example, we find that DtACI and AgACI successfully correct for these shifts, while MVP provides inconsistent coverage across the four examples. This contrasts sharply with the baseline method that holds $\alpha_t = \alpha$ fixed, which, depending on the example, undergoes large swings (e.g. top-left panel) or displays a systematic bias (e.g. bottom-left panel) away from the target coverage frequency.

\section{Acknowledgments}

E.C. was supported by the Office of Naval Research grant N00014-20-1-2157, the National Science Foundation grant DMS-2032014, the Simons Foundation under award 814641, and the ARO grant 2003514594. I.G. was supported under the Office of Naval Research grant N00014-20-1-2157 and the Simons Foundation under award 814641 as well as by the Overdeck Fellowship Fund. The authors thank Elad Hazan for useful pointers to the online learning literature and John Cherian for helpful comments on an earlier version of this work.

\newpage

\bibliography{FullAdaptiveConformalBib}

\newpage

\appendix

\section{Appendix}

\subsection{Detailed description of multivalid conformal prediction}\label{sec:MVP}

The generic version of multivalid conformal prediction (MVP) proposed by \citet{Bastani2022} is designed to give simultaneous coverage over a collection of subsets of the covariate space. This is accomplished by constructing prediction sets of the form $\{y : S_t(X_t,y) \leq q\}$, where $q$ is chosen from a set of candidate values based off of their performance on historical data. To make this method comparable to DtACI and AgACI we implement a modified version that does not consider any subsets of the covariate space and treats $1-\beta_t$ as the conformity score. More specifically, our implementation outputs prediction sets of the form $\{y : 1-\beta_t(y) \leq q\}$, where $q$ is chosen using the MVP algorithm and 
\[
\beta_t(y) := \max\{0 \leq \beta \leq 1 : S_t^y \leq \text{Quantile}(1-\beta, \mathcal{D}^y_t)\}.  
\] 
With this construction the value of $1-q$ output by MVP is exactly anologous to the values $\alpha_t$ output by DtACI and AgACI.

The full details of our implementation are given in Algorithm \ref{alg:MVP} below. Following the original implementation of MVP in \citet{Bastani2022} we take our hyperparameters to be $m = 40$, $\eta = \sqrt{\frac{\log(m)}{4.2 m}}$, and $r = 800000$. Finally, in what follows $f(\cdot)$ refers to the function $f(n) := \sqrt{(n+1)\log(n+2)^{2}}$
 
\begin{algorithm}
 \KwData{Observed values $\{\beta_t\}_{1 \leq t \leq T}$, target coverage level $\alpha$, number of candidate thresholds $m$, hyperparameters $\eta$, $r$.}
 \For{$t=1,2,\dots,T$}{
    \For{$i=0,1,\dots,m-1$}{
 	$n^i_t \leftarrow \sum_{s<t} \bone\{i/m \leq q_s < (i+1)/m \text{ or } i=m-1, q_s=1 \}$;\\
    $V_t^i \leftarrow \sum_{s<t} \bone\{i/m \leq q_s < (i+1)/m \text{ or } i=m-1, q_s=1\}(\alpha - \bone\{\beta_s < 1-q_s\} )$;\\
    $C^i_t = (\exp(\eta V_t^i/f(n_t^i))-\exp(-\eta V_t^i/f(n_t^i)))/f(n_t^i)$;
    }
    \If{$C^i_t > 0$ for all $i \in [m]$}{
    Output $q_t = 0$;
    }\ElseIf{$C^i_t < 0$ for all $i \in [m]$}{
    Output $q_t = 1$;
    }\Else{
    Set $i^* \in [m-1]$ to be the minimum index such that $C^{i^*}_{t}C^{i^*+1}_{t} \leq 0$;\\
    $p_t \leftarrow |C^{i^*+1}_{t}|/(|C^{i^*+1}_{t}| + |C^{i^*}_{t}|)$, with the convention $0/0=1$;\\
    Output $q_t = i^*/m - 1/(rm)$ with probability $p_t$ and $q_t = i^*/m$ with probability $1-p_t$;
    }
    Output prediction set $\{y : 1-\beta_t(y) \leq q_t\}$;
 }
 
 \caption{Modified version of the MVP algorithm of \cite{Bastani2022}.}
 \label{alg:MVP}
\end{algorithm}

\subsection{Details of the block bootstrap for Section \ref{sec:cond_on_alphat_cov}}

Algorithm \ref{alg:block_boot} below outlines the block bootstrap procedure used to generate the error bars in Figures \ref{fig:cond_alphat_cov} and \ref{fig:cond_alphat_cov_unnormalized}.  Both figures were generated using the choices $M=100$ and $b=100$.

\begin{algorithm}
 \KwData{Sequence of stock returns $\{R_t\}_{1 \leq t \leq T}$, number of bootstrap samples $M$, block size $b$, partition $B_1,\dots,B_k$ of $[0,1]$.}
	Define the blocks of data $D_i := \{R_{ib+1},\dots,R_{(i+1)b\}}\}$, $0 \leq i \leq  T/b - 1 $\;
 \For{$j=1,2,\dots,M$}{
 	Sample $i^j_1,\dots,i^j_{T/b} \overset{\text{i.i.d.}}{\sim} \text{Unif}(\{0,\dots,T/b-1\})$\;
 	Run the procedure outlined in Section \ref{sec:stocks} on dataset $\{D_{i^j_1},\dots,D_{i^j_{T/b}}\}$ to obtain sequences $\{\bar{\alpha}_t^j\}_{1 \leq t \leq T}$ and  $\{\text{err}_t^j\}_{1\leq t \leq T}$\;
 	For $1 \leq \ell \leq k$ compute
 	\[
 	\text{CondCoverage}^j_{\ell} := \frac{1}{ |\{t : \bar{\alpha}^j_t \in B_{\ell}\}|  } \sum_{t : \bar{\alpha}^j_t \in B_{\ell}} \text{err}^j_t.
 	\]
 }
 For all $1 \leq \ell \leq k$ output the 0.05 and 0.95 empirical quantiles of $\{\text{CondCoverage}^j_{\ell}\}_{1 \leq j \leq M}$.
  \caption{Block bootstrap}
 \label{alg:block_boot}
\end{algorithm}

\subsection{Proofs for Section \ref{sec:reg_bounds}}

This section contains the proofs of Lemmas \ref{lem:adapt_reg_bound}  and \ref{lem:reg_ogd} as well as of Theorem \ref{thm:dynamic_reg}. In addition we prove a modified version of Theorem \ref{thm:dynamic_reg} that allows $\eta$ to vary over time. We begin by proving the results stated in Section \ref{sec:reg_bounds}.

\subsubsection{Proof of Lemma \ref{lem:adapt_reg_bound}} 
\begin{proof}
We follow the calculations of \cite{Gradu2021}. Let $\ell(\beta_t) := (\ell(\beta_t,\alpha_t^1),\dots,\ell(\beta_t,\alpha_t^k))$, $\ell(\beta_t)^2 := (\ell(\beta_t,\alpha_t^1)^2,\dots,\ell(\beta_t,\alpha_t^k)^2)$ and $p_t := (p_t^1,\dots,p_t^k)$. By construction notice that $W_{t+1} := \sum_{i=1}^k w_{t+1}^i = \bar{W}_t$. Thus, using the inequalities $\exp(-x) \leq 1-x+x^2$ and $1-y \leq \exp(-y)$ we find that 
\[
\frac{W_{t+1}}{W_t} = \sum_{i=1}^k p_t^i \exp(-\eta \ell(\beta_t,\alpha_t^i)) \leq \exp(-\eta p_t^T  \ell(\beta_t) + \eta^2 p_t^T  \ell(\beta_t)^2),
\]
and thus inductively 
\[
W_{s+1}/W_r \leq \exp\left(-\sum_{t=r}^s \eta p_t^T  \ell(\beta_t) + \eta^2 p_t^T  \ell(\beta_t)^2\right).
\]
On the other hand, for any $i$, $w^i_{t+1} \geq w^i_t(1-\sigma)\exp(-\eta \ell(\beta_t,\alpha_t^i))$ which gives
\[
\frac{W_{s+1}}{W_r} \geq \frac{w^i_{s+1} }{W_r} \geq (1-\sigma)^{|I|}\exp\left(-\sum_{t=r}^s \eta \ell(\beta_t,\alpha_t^i)\right)p_r^i \geq (1-\sigma)^{|I|}\exp\left(-\sum_{t=r}^s \eta \ell(\beta_t,\alpha_t^i)\right) \frac{\sigma}{k}.
\]
Combining these two inequalities and taking a logarithm yields 
\[
\log(\sigma/ k) + |I|\log(1-\sigma) - \sum_{t=r}^s \eta \ell(\beta_t,\alpha_t^i) \leq -\sum_{t=r}^s \eta p_t^T  \ell(\beta_t) + \eta^2 p_t^T  \ell(\beta_t)^2.
\]
Finally, since $\sigma \leq 1/2$ we may use the inequality $\log(1-\sigma) \geq -2\sigma$ to get the final result
\[
 \sum_{t=r}^s \mme[\ell(\beta_t,\alpha_t)] \leq   \sum_{t=r}^s \ell(\beta_t,\alpha_t^i) + \eta\sum_{t=r}^s \mme[\ell(\beta_t,\alpha_t)^2]  +  \frac{1}{\eta}\left(\log(k/\sigma\right) + |I|2\sigma ).
\]
\end{proof}

\subsubsection{Proof of Lemma \ref{lem:reg_ogd}}
\begin{proof}
Lemma 4.1 in \cite{Gibbs2021} shows that for all values $t$, $\alpha_t^i \in [-\gamma_i,1+\gamma_i]$. Since $\beta_t \in [0,1]$ we then also have that $\ell_t(\beta_t,\alpha_t^i) \leq \max\{\alpha,1-\alpha\} | \beta_t - \alpha^i_t| \leq 1+\gamma_i$. Plugging this fact into Theorem 10.1 of \cite{Hazan2019} gives the result.
\end{proof}

\subsubsection{Proof of Theorem \ref{thm:dynamic_reg}}
\begin{proof}
Fix any $i \in [k]$ and write 
\begin{align*}
 & \sum_{t=r}^s \mme[\ell(\beta_t,\alpha_t)]  -  \sum_{t=r}^s \ell(\beta_t,\alpha^*_t)\\
 & =  \bigg(\sum_{t=r}^s \mme[\ell(\beta_t,\alpha_t)] -  \sum_{t=r}^s \ell(\beta_t,\alpha^i_t) \bigg)  +  \bigg(\sum_{t=r}^s \ell(\beta_t,\alpha^i_t)  -  \sum_{t=r}^s \ell(\beta_t,\alpha^*_t) \bigg).
\end{align*}
Applying Lemma \ref{lem:adapt_reg_bound} to the first term and Lemma \ref{lem:reg_ogd} to the second term gives
\begin{align*}
 \sum_{t=r}^s \mme[\ell(\beta_t,\alpha_t)]  -  \sum_{t=r}^s \ell(\beta_t,\alpha^*_t) \leq  & \eta\sum_{t=r}^s \mme[\ell(\beta_t,\alpha_t)^2]  +  \frac{1}{\eta}\left(\log(|E|/\sigma\right) + |I|2\sigma )\\
  & + \frac{3}{2\gamma_i} (1+\gamma_i)^2 \left(\sum_{t=r+1}^s |\alpha^*_t - \alpha^*_{t-1}| + 1 \right) + \frac{1}{2}\gamma_i |I|.
\end{align*}
Now, there are two cases. If 
\begin{equation}\label{eq:path_versus_gamma}
\sqrt{ \frac{\sum_{t=r+1}^s |\alpha^*_t - \alpha^*_{t-1}| + 1}{|I|}} \geq \gamma_1
\end{equation}
then since $\sqrt{\frac{\sum_{t=r+1}^s |\alpha^*_t - \alpha^*_{t-1}| + 1 }{|I|}} \leq \sqrt{1+1/|I|} \leq \gamma_{k}$ we may find a value $\gamma_i$ such that 
\[
\sqrt{ \frac{\sum_{t=r+1}^s |\alpha^*_t - \alpha^*_{t-1}| + 1}{|I|}}  \leq \gamma_i \leq 2\sqrt{ \frac{\sum_{t=r+1}^s |\alpha^*_t - \alpha^*_{t-1}| + 1}{|I|}} .
\]
Plugging this value into the previous expression gives the desired result. Otherwise if (\ref{eq:path_versus_gamma}) does not hold, then we may simply plug in $\gamma_1$ for $\gamma_i$.
\end{proof}

\subsubsection{Results for Variable $\eta$}

We conclude this section by stating and proving a modified version of Theorem \ref{thm:dynamic_reg} that allows $\eta$ to vary over time. In particular, we consider a modified version of DtACI in which the update for $\bar{w}^i_t$ is replaced by $\bar{w}^i_t \leftarrow w^i_t \exp(-\eta_t\ell_t(\beta_t,\alpha^i_t))$. Now, recall that our regret bounds consider a target interval length $|I| = L$. Then, Theorem \ref{thm:dymanic_eta_regret} below shows that if the variation in $\eta_t$ is of order $1/\sqrt{L}$, DtACI will have dynamic regret of size $O(\sqrt{\log(L)/{L}} +  \max\{\gamma_1,\sqrt{\frac{1}{L}\sum_{t=r+1}^s |\alpha^*_t - \alpha^*_{t-1}|}\})$. 

The primary case of interest is that in which $\eta_t = \sqrt{\frac{\log(L\cdot k) + 2}{\sum_{s=t-L+1}^{t} \mme[\ell(\beta_t,\alpha_t)^2]}}$ and $\sigma_t = \frac{1}{2L}$. Here the assumption that $\eta_t$ has small variability is motivated by a time-series model in which $(\beta_t,(\alpha^t_i,w_t^i)_{i \in [k]}))$ has a stationary distribution. For example, the original work of \cite{Gibbs2021} considers a setting in which $\{(X_t,Y_t)\}$ follows a hidden Markov model and $S_t^y = S(X_t,y)$ and $\mathcal{D}_t^y = \mathcal{D}$ are fixed quantities that do not depend on $t$. In this set-up, it follows that $(\beta_t,(\alpha^t_i,w_t^i)_{i \in [k]})$ forms a Markov chain and thus under reasonable mixing assumptions on $(X_t,Y_t)$ one can expect $\frac{1}{L}\sum_{s=t-L+1}^{t} \mme[\ell(\beta_t,\alpha_t)^2]$ to have variations of order $1/\sqrt{L}$.

\begin{theorem}\label{thm:dymanic_eta_regret}
Let $L \in \mmn$ denote a fixed target interval length and $I=[r,s]$ be any interval of length $L$ with $r>L$. For all $t>L$ set $\eta_t := \sqrt{\frac{\log(L\cdot k) + 2}{\sum_{s=t-L+1}^{t} \mme[\ell(\beta_t,\alpha_t)^2]}}$ and $\sigma = 1/(2L)$. Let $\{\alpha_t^i, 
p_t^i\}_{t \in [T], i \in [k]}$ be generated by a modified version of Algorithm \ref{alg:faci} in which the update for $\bar{w}^i_t$ is replaced by $\bar{w}^i_t \leftarrow w^i_t \exp(-\eta_t\ell_t(\beta_t,\alpha^i_t))$. Assume that 
\[
\frac{1}{L\eta_s}\sum_{t=r}^s |\eta_t - \eta_s| , \frac{1}{L\eta_s}\sum_{t=r}^s |\eta_t^2 - \eta^2_s| \leq O\left(\frac{1}{\sqrt{L}} \right).
\]
Then, under the conditions of Theorem \ref{thm:dynamic_reg}
\begin{align*}
 \frac{1}{L} \sum_{t=r}^s\mme[ \ell(\beta_t,\alpha_t)] - & \frac{1}{L} \sum_{r=1}^s \ell(\beta_t,\alpha_t^*) \leq  2\sqrt{\frac{ \log(L\cdot k) + 2 }{L}} \sqrt{\frac{1}{L}\sum_{t=r}^s \mme[\ell(\beta_t,\alpha_t)^2]}\\
 & \ \ \ \ \ \ \ \ \ \ \ \ \ \ \ \ \ \ \ \  + 4(1+\gamma_{\textup{max}})^2\max\left\{\sqrt{ \frac{ \sum_{t=r+1}^s |\alpha^*_t - \alpha^*_{t-1}|+1}{L}}, \gamma_1\right\}\\
 & \ \ \ \ \ \ \ \ \ \ \ \ \ \ \ \ \ \ \ \  + O\left(\frac{1}{\sqrt{L}} \right)   \\
 & \ \ \ \ \ \ \ \ \ \ \ \ \ \ \ \ \ \ = O\left( \sqrt{\frac{\log(L)}{L}}\right) + O\left( \max\{\gamma_1,\frac{1}{L}\sum_{t=r+1}^s |\alpha^*_t - \alpha^*_{t-1}|\} \right),
\end{align*}
where the expectation is over the randomness in DtACI and the data $\beta_1,\dots,\beta_T$ can be viewed as fixed.
\end{theorem}

\begin{proof}
Proceeding identically to the proof of Lemma \ref{lem:adapt_reg_bound} we have that for any $i \in [k]$
\[
 \sum_{t=r}^s \eta_t\mme[\ell(\beta_t,\alpha_t)] \leq   \sum_{t=r}^s \eta_t\ell(\beta_t,\alpha_t^i) + \sum_{t=r}^s \eta_t^2 \mme[\ell(\beta_t,\alpha_t)^2]  +  \log\left(k/\sigma\right) + |I|2\sigma 
\]
Now, by Lemma 4.1 of \cite{Gibbs2021} we know that $\alpha_t^i \in [-\gamma_i,1+\gamma_i]$ and thus that  $\ell(\beta_t,\alpha^i_t),\ell(\beta_t,\alpha_t) \leq 1+\gamma_{\text{max}}$. Hence, 
\begin{align*}
 \eta_s \sum_{t=r}^s\mme[\ell(\beta_t,\alpha_t)] &  \leq \eta_{s} \sum_{t=r}^s \ell(\beta_t,\alpha_t^i) + \eta_s^2 \sum_{t=r}^s \mme[\ell(\beta_t,\alpha_t)^2]  +  \log(k/\sigma) + |I|2\sigma\\
& \quad \quad + 2(1+\gamma_{\text{max}})\sum_{t=r}^s |\eta_t - \eta_s| + (1+\gamma_{\text{max}})^2\sum_{t=r}^s |\eta_t^2 - \eta^2_s|.
\end{align*}
The remainder of the proof is identical to that of Theorem \ref{thm:dynamic_reg}.

\end{proof}

\subsection{Proofs for Section \ref{sec:coverage_bounds}}

\subsubsection{Proof of Proposition \ref{prop:loss_connection}}

\begin{proof}
For simplicity we will only prove the case where $\tau > \alpha^*$. The case where $\tau \leq \alpha^*$ is identical. By a direct computation we have that 
\begin{align*}
& \mme[\ell(\beta,\tau)] - \mme[\ell(\beta,\alpha^*)]\\
& = \mme[\alpha(\beta -\tau)\bone_{\beta \geq \tau}] + \mme[(1-\alpha)( \tau - \beta)\bone_{\beta <\tau}] -  \mme[\alpha(\beta - \alpha^*)\bone_{\beta \geq \alpha^*}] - \mme[(1-\alpha)( \alpha^* - \beta)\bone_{\beta < \alpha^*}] \\
& = -\mme[\alpha (\tau - \alpha^*) \bone_{\beta \geq \tau}] + \mme[(1-\alpha) (\tau - \alpha^*)\bone_{\beta < \alpha^*}] + \mme[((1-\alpha)( \tau - \beta) - \alpha(\beta - \alpha^*))\bone_{ \alpha^* \leq \beta <\tau}]\\
& = -\mme[\alpha  (\tau - \alpha^*) \bone_{\beta \geq \alpha^*}] + \mme[(1-\alpha) (\tau - \alpha^*)\bone_{\beta < \alpha^*}] + \mme[\alpha (\tau - \alpha^*) \bone_{\alpha^* \leq \beta < \tau}]\\
& \ \ \ \ \  + \mme[(1-\alpha) (\tau - \alpha^*) \bone_{\alpha^* \leq \beta < \tau}] - \mme[(\beta - \alpha^*)\bone_{\alpha^* \leq \beta < \tau}]\\
& = -\alpha(1-\alpha) (\tau - \alpha^*) + \alpha(1-\alpha) (\tau - \alpha^*) + \mme[ (\tau - \alpha^*) \bone_{\alpha^* \leq \beta < \tau}] - \mme[(\beta - \alpha^*)\bone_{\alpha^* \leq \beta < \tau}]\\
& = \mme[(\tau- \beta) \bone_{\alpha^* \leq \beta < \tau}].
\end{align*}
This proves the first part of Proposition \ref{prop:loss_connection}. For the second part note that if $\beta$ has a density on $[0,1]$ that is lower bounded by $p$, then  
\[
 \mme[(\tau - \beta) \bone_{\alpha^* \leq \beta < \tau}], \  \mme[(\beta - \tau)  \bone_{\tau \leq \beta < \alpha^*}] \geq \int_{0}^{ |\tau - \alpha^*|}xpdx =p\frac{(\tau - \alpha^*)^2}{2}.
\]
\end{proof}

\subsection{Proofs for Section \ref{sec:coverage_bounds}}

\subsubsection{Proof of Theorem \ref{thm:long_term_coverage}}

\begin{proof}
Let
\[
\tilde{\alpha}_t := \sum_i \frac{p_t^i \alpha_t^i}{\gamma_i}
\]
and observe that
\begin{align*}
\tilde{\alpha}_t & = \sum_i \frac{p_t^i (\alpha_{t+1}^i - \gamma_i (\alpha - \text{err}_t^i))}{\gamma_i} =  \sum_i \frac{p_t^i \alpha_{t+1}^i}{\gamma_i}  + \sum_i p_t^i  (\text{err}_t^i - \alpha)\\
& = \tilde{\alpha}_{t+1} + \sum_i \frac{(p_t^i - p_{t+1}^i) \alpha_{t+1}^i}{\gamma_i} + \sum_i p_t^i (\text{err}_t^i - \alpha) .
\end{align*}
Thus,
\begin{equation}\label{eq:err_t_rewrite}
\mme[\text{err}_t] - \alpha = \tilde{\alpha}_t - \tilde{\alpha}_{t+1} + \sum_i \frac{(p_{t+1}^i - p_{t}^i) \alpha_{t+1}^i}{\gamma_i}.
\end{equation}
Now, for ease of notation let $W_t := \sum_i w_t^i$ and $\tilde{p}^i_{t+1} := \frac{p_t^i\exp( - \eta_t \ell(\beta_t,\alpha_t^i))}{\sum_{i'} p_t^{i'}\exp( - \eta_t \ell(\beta_t,\alpha_t^{i'}))  }$. Recall that by definition,
\[
p_{t+1}^i  =  \frac{w_{t+1}^i}{\sum_{i'} w_{t+1}^{i'}} = (1-\sigma_t) \tilde{p}^i_{t+1} + \frac{\sigma_t}{k}.
\]
Now, by a direct computation
\begin{align*}
 \tilde{p}^i_{t+1} - p^i_t & = \frac{p_t^i\exp( - \eta_t \ell(\beta_t,\alpha_t^i))}{\sum_{i'} p_t^{i'}\exp( - \eta_t \ell(\beta_t,\alpha_t^{i'}))  }  - p_t^i\\
 & = p_t^i \frac{\exp( - \eta_t \ell(\beta_t,\alpha_t^i)) -    \sum_{i'} p_t^{i'}\exp( - \eta_t \ell(\beta_t,\alpha_t^{i'}))}{  \sum_{i'} p_t^{i'}\exp( - \eta_t \ell(\beta_t,\alpha_t^{i'}))} \\
 & = p_t^i  \frac{    \sum_{i'} p_t^{i'}(\exp( - \eta_t \ell(\beta_t,\alpha_t^{i}))  - \exp( - \eta_t \ell(\beta_t,\alpha_t^{i'})))}{  \sum_{i'} p_t^{i'}\exp( - \eta_t \ell(\beta_t,\alpha_t^{i'}))} \\
  & = p_t^i \frac{    \sum_{i'} p_t^{i'} \exp(  -\eta_t \ell(\beta_t,\alpha_t^{i'})) (\exp(  \eta_t \ell(\beta_t,\alpha_t^{i'}) - \eta_t \ell(\beta_t,\alpha_t^i))  -1)}{  \sum_{i'} p_t^{i'}\exp( - \eta_t \ell(\beta_t,\alpha_t^{i'}))} \\
  & =  p_t^i  \sum_{i' } \tilde{p}_{t+1}^{i'} (\exp(  \eta_t \ell(\beta_t,\alpha_t^{i'}) - \eta_t \ell(\beta_t,\alpha_t^{i}) )  -1).
\end{align*}
By Lemma 4.1 of \cite{Gibbs2021} we know that $\alpha_t^i \in [-\gamma_i,1+\gamma_i]$ and thus that $ |\ell(\beta_t,\alpha_t^{i'}) -  \ell(\beta_t,\alpha_t^{i})| \leq \max\{\alpha,1-\alpha\}|\alpha_t^{i'} - \alpha_t^{i}| \leq 1+2\gamma_{\text{max}}$. Hence, by the mean value theorem
\[
| \exp(  \eta_t \ell(\beta_t,\alpha_t^{i'}) - \eta_t \ell(\beta_t,\alpha_t^{i}) )  -1| \leq \eta_t(1+2\gamma_{\text{max}})\exp(\eta_t(1+2\gamma_{\text{max}})),
\]
and thus also
\[
| \tilde{p}^i_{t+1} - p^i_t| \leq p_t^i\eta_t(1+2\gamma_{\text{max}})\exp(\eta_t(1+2\gamma_{\text{max}})).
\]
Applying  Lemma 4.1 of \cite{Gibbs2021} again we conclude that 
\begin{align*}
\left|\sum_i \frac{(p_{t+1}^i - p_{t}^i) \alpha_{t+1}^i}{\gamma_i} \right| & \leq (1-\sigma_t)\sum_i \left|\frac{(\tilde{p}_{t+1}^i - p_{t}^i) \alpha_{t+1}^i}{\gamma_i}\right| + \sigma_t \sum_i\left| \frac{(1/k - p_{t}^i) \alpha_{t+1}^i}{\gamma_i} \right|\\
& \leq \frac{\eta_t(1+2\gamma_{\text{max}})^2}{\gamma_{\text{min}}}\exp(\eta_t(1+2\gamma_{\text{max}})) + 2\sigma_t \frac{1+\gamma_{\text{max}}}{\gamma_{\text{min}}}.
\end{align*}
So, taking a sum over $t$ in equation \ref{eq:err_t_rewrite} gives the inequality
\[
\left| \frac{1}{T} \sum_{t=1}^T \mme[\text{err}_t] - \alpha \right| \leq  \frac{|\tilde{\alpha}_{1} - \tilde{\alpha}_{T+1}|}{T} +  \frac{(1+2\gamma_{\text{max}})^2}{\gamma_{\text{min}}}   \frac{1}{T} \sum_{t=1}^T \eta_t e^{\eta_t(1+2\gamma_{\text{max}})}+ 2  \frac{1+\gamma_{\text{max}}}{\gamma_{\text{min}}}\frac{1}{T} \sum_{t=1}^T \sigma_t. 
\]
Applying Lemma 4.1 of \cite{Gibbs2021} one final time gives $\gamma_{\text{min}}\tilde{\alpha}_{t} \in [-\gamma_{\text{max}},1+\gamma_{\text{max}}]$ and thus $|\tilde{\alpha}_{1} - \tilde{\alpha}_{T+1}| \leq (1+2\gamma_{\text{max}})/\gamma_{\text{min}}$. Plugging this into the previous expression gives the desired upperbound. 

Finally, if $\eta_t, \sigma_t \to 0$ then taking limits gives the desired equalities
\[
\lim_{T \to \infty} \left| \frac{1}{T} \sum_t \text{err}_t - \alpha \right| = \lim_{T \to \infty}  \left| \frac{1}{T} \sum_t \mme[\text{err}_t] - \alpha \right| = 0,
\]
where the first equality follows from the law of large numbers.

\end{proof}

\subsection{Additional figures}

We begin by giving three figures showing the results for Section \ref{sec:real_data} in the case where we use a variable value for $\eta$, given by $\eta = \eta_t = \sqrt{\frac{\log(500 \cdot k) + 2}{\sum_{s=t-501}^{t} \mme[\ell(\beta_t,\alpha_t )^2]}}$. We see that the results are nearly exactly identical to those obtained with the fixed heurisitic choice of $\eta$.

\begin{figure}[H]
\begin{centering}
\includegraphics[scale=0.16]{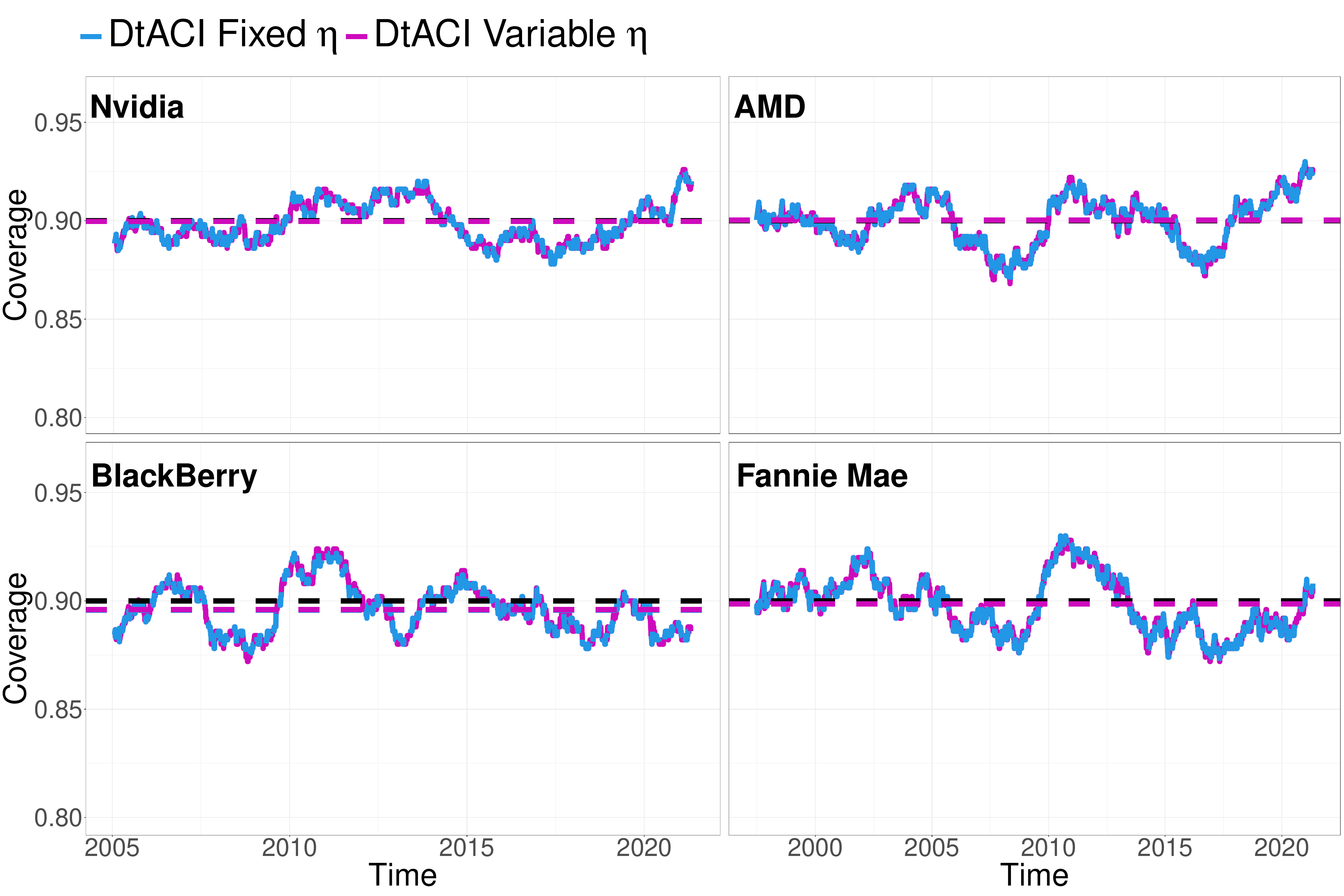}
\caption{Estimation of stock market volatility using conformity score $S_t(v) = |v - (\hat{\sigma}^t_t)^2|/(\hat{\sigma}^t_t)^2$ and either $\eta$ fixed at 2.8 (blue) or a variable value (purple) of $\eta_t = \sqrt{\frac{\log(500 \cdot k) + 2}{\sum_{s=t-501}^{t} \mme[\ell(\beta_t,\alpha_t )^2]}}$. Solid lines show the local coverage level for $\bar{\alpha}_t$, while dashed lines indicate the global coverage frequency over all time steps. The black dashed lines indicates the target level of $1-\alpha = 0.9$.}
\label{fig:volatility_cov_dynamic_eta_normalized}
\end{centering}
\end{figure} 

\begin{figure}[H]
\begin{centering}
\includegraphics[scale=0.16]{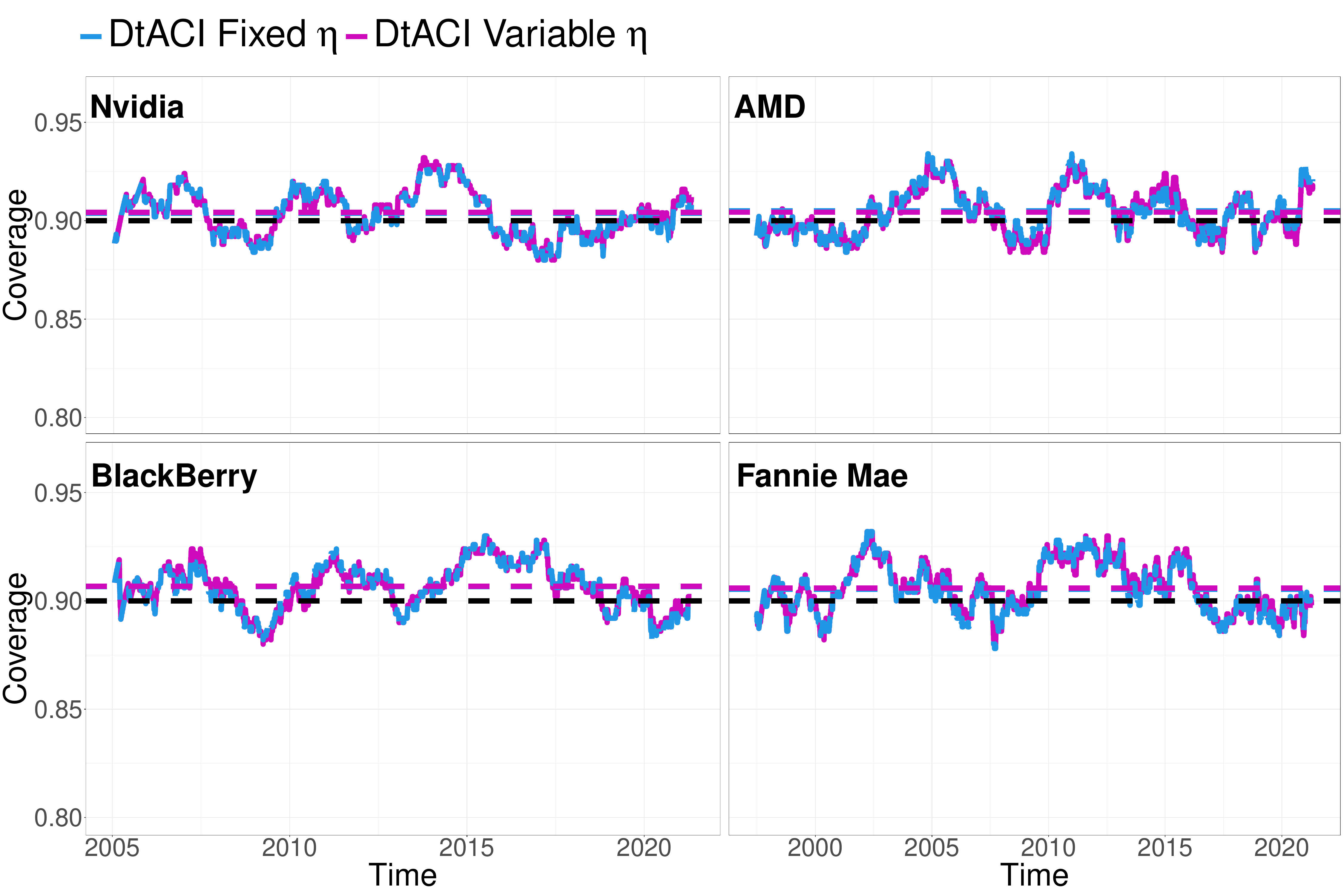}
\caption{Estimation of stock market volatility using conformity score $S_t(v) = |v - (\hat{\sigma}^t_t)^2|$ and either $\eta$ fixed at 2.8 (blue) or a variable value (purple) of $\eta_t = \sqrt{\frac{\log(500 \cdot k) + 2}{\sum_{s=t-501}^{t} \mme[\ell(\beta_t,\alpha_t )^2]}}$. Solid lines show the local coverage level for $\bar{\alpha}_t$, while dashed lines indicate the global coverage frequency over all time steps. The black dashed lines indicates the target level of $1-\alpha = 0.9$.}
\label{fig:volatility_cov_dynamic_eta_unnormalized}
\end{centering}
\end{figure} 

\begin{figure}[H]
\begin{centering}
\includegraphics[scale=0.16]{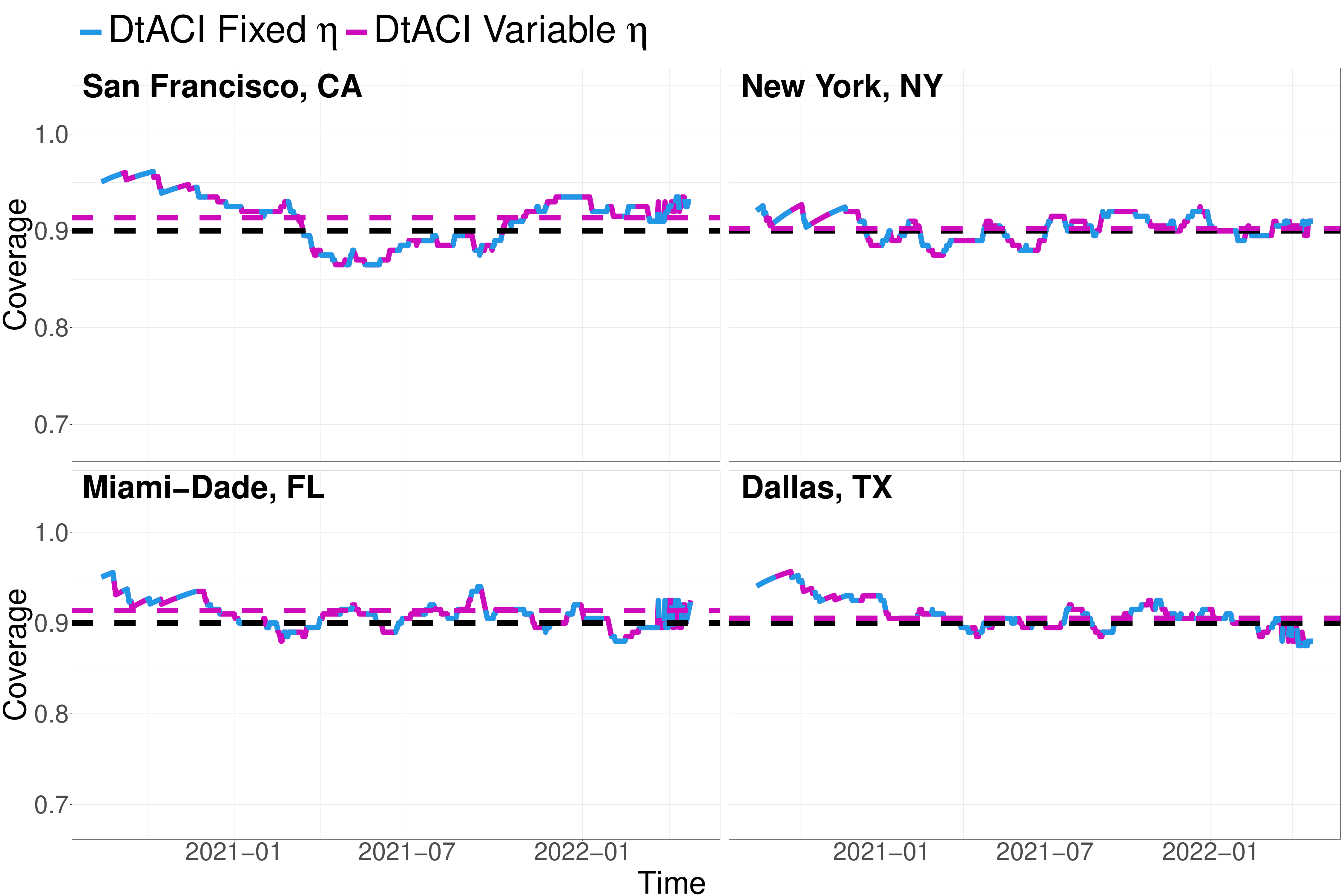}
\caption{Coverage results for the prediction of county-level COVID-19 case counts using either $\eta$ fixed at 2.8 (blue) or a variable value (purple) of $\eta_t = \sqrt{\frac{\log(500 \cdot k) + 2}{\sum_{s=t-500}^{t-1} \mme[\ell(\beta_t,\alpha_t )^2]}}$. Solid lines show the local coverage level for $\bar{\alpha}_t$, while dashed lines indicate the global coverage frequency over all time steps. The black dashed lines indicates the target level of $1-\alpha = 0.9$.}
\label{fig:covid_cov_dynamic_eta}
\end{centering}
\end{figure} 

Our next figure shows Q-Q plots comparing the quantiles of the local coverage gaps realized by DtACI against those for an i.i.d. Bernoulli(0.9) sequence for the volatility dataset of Section \ref{sec:stocks} and the unnormalized conformity score $S_t(v) = |v - (\hat{\sigma}^t_t)^2|$.

\begin{figure}[H]
\begin{centering}
\includegraphics[scale=0.16]{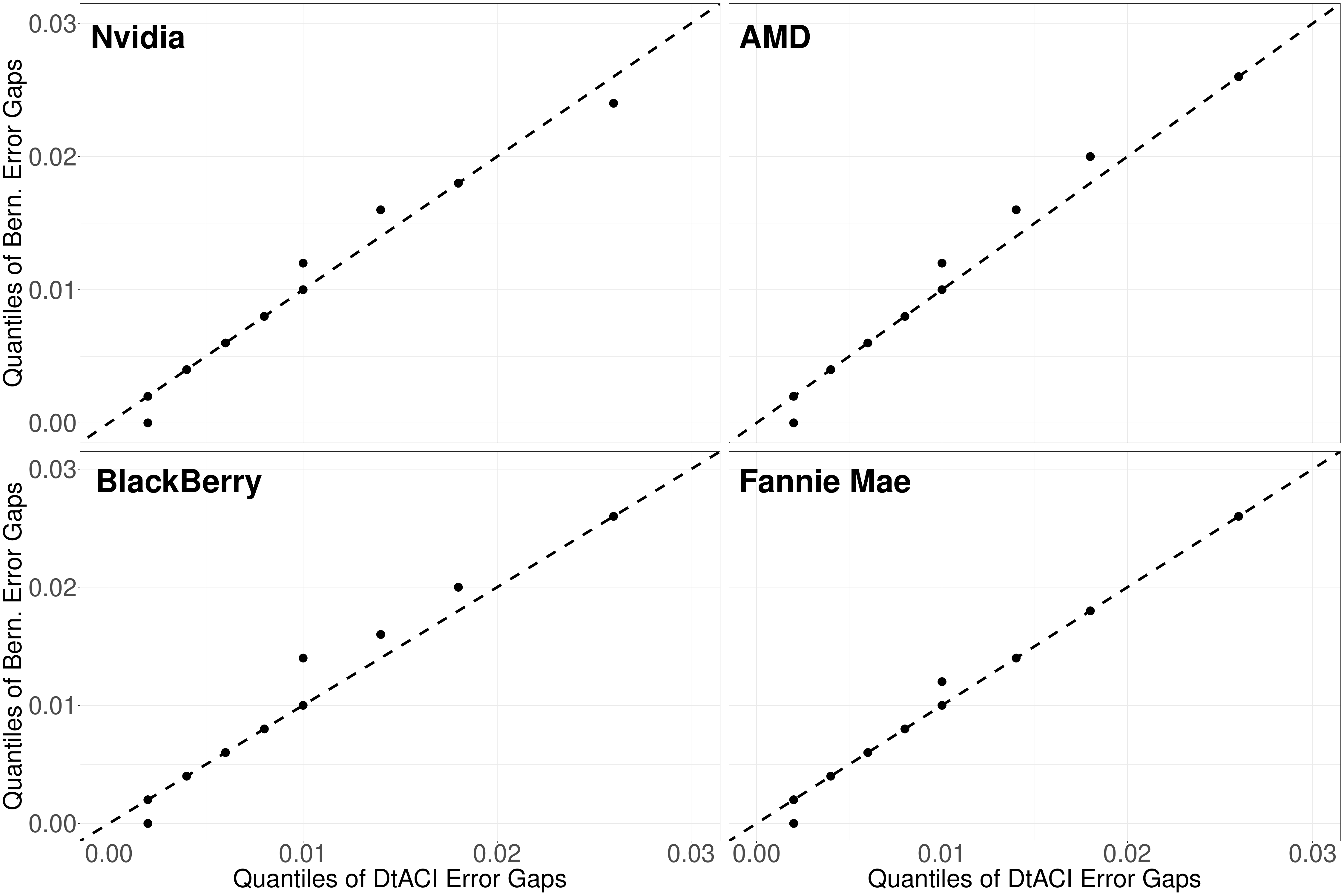}
\caption{Q-Q plots comparing the distribution of local coverage gaps obtained by an i.i.d. Bernoulli($\alpha$) against those realized by DtACI. Results for DtACI were generated with conformity score $S_t(v) = |v - (\hat{\sigma}^t_t)^2|$. The dashed line indicates the ideal situation of exact equality.}
\label{fig:bern_comp_bad}
\end{centering}
\end{figure}

The next figure shows the empirical conditional coverage (\ref{eq:cond_coverage}) for the estimation of stock market volatility with the unnormalized conformity score. As for the normalized conformity score we observe conditional coverages close to the target level across all values of $\bar{\alpha}_t$. 

\begin{figure}[H]
\begin{centering}
\includegraphics[scale=0.16]{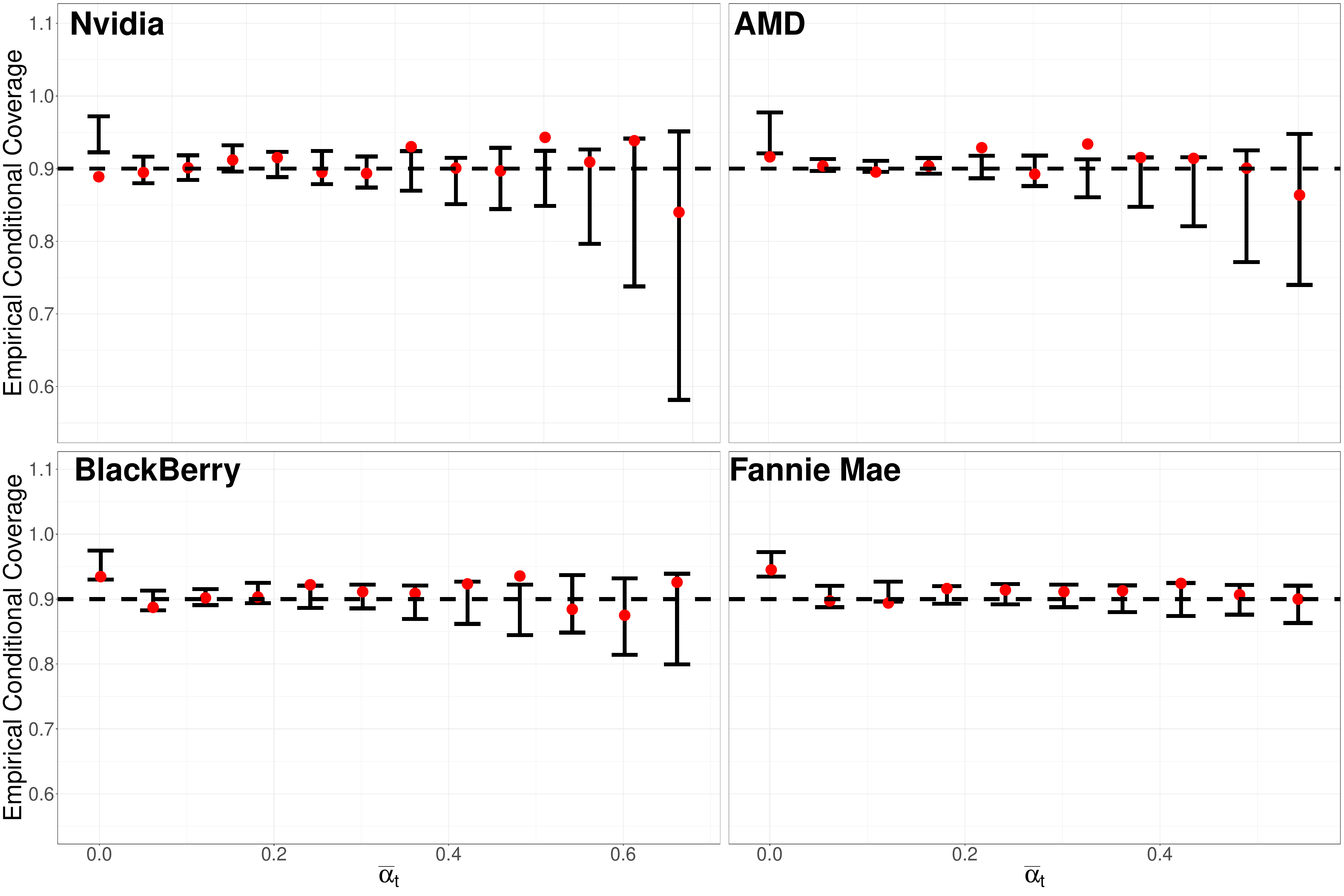}
\caption{Empirical conditional coverage of $\bar{\alpha}_t$ for the estimation of stock market volatility with conformity score $S_t(v) = |v - (\hat{\sigma}^t_t)^2|$. Red points show the empirical conditional coverage given $\bar{\alpha}_t \in B_i$ with error bars indicating the corresponding 0.05 and 0.95 quantiles across 100 block bootstrap resamples of the data $\{(X_i,Y_i)\}$ with block-size set to 100. Black dashed lines shows the target level of $1-\alpha = 0.9$. }
\label{fig:cond_alphat_cov_unnormalized}
\end{centering}
\end{figure} 

Finally, Figures \ref{fig:stock_prices} and \ref{fig:covid_case_counts} show the stock prices and the COVID-19 case counts for the datasets considered in Section \ref{sec:real_data}.

\begin{figure}[H]
  \centering
  \includegraphics[scale=0.16]{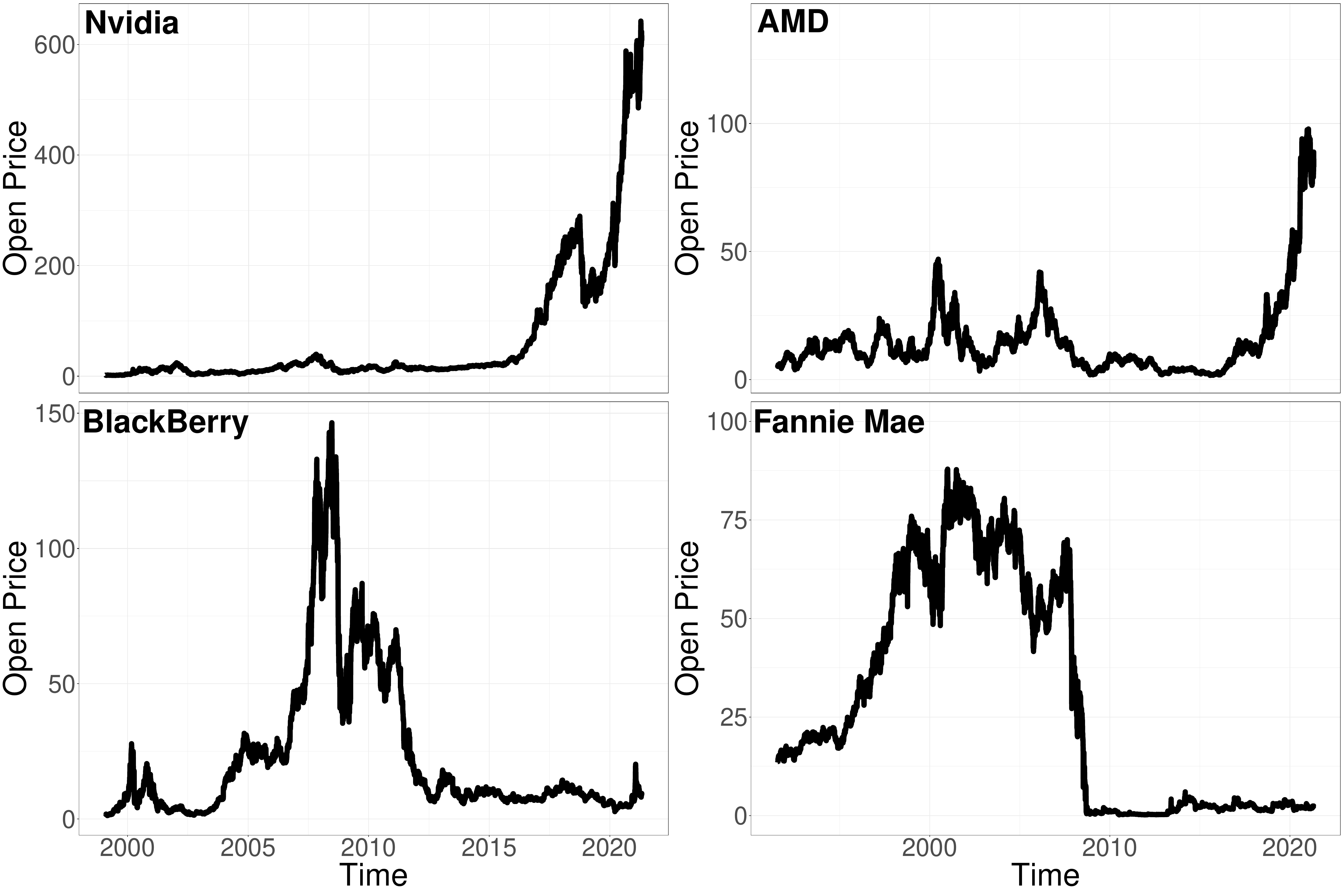}
  \caption{Daily open prices for the four stocks considered in Section \ref{sec:stocks}.}
  \label{fig:stock_prices}
\end{figure}

\begin{figure}[H]
\begin{centering}
\includegraphics[scale=0.16]{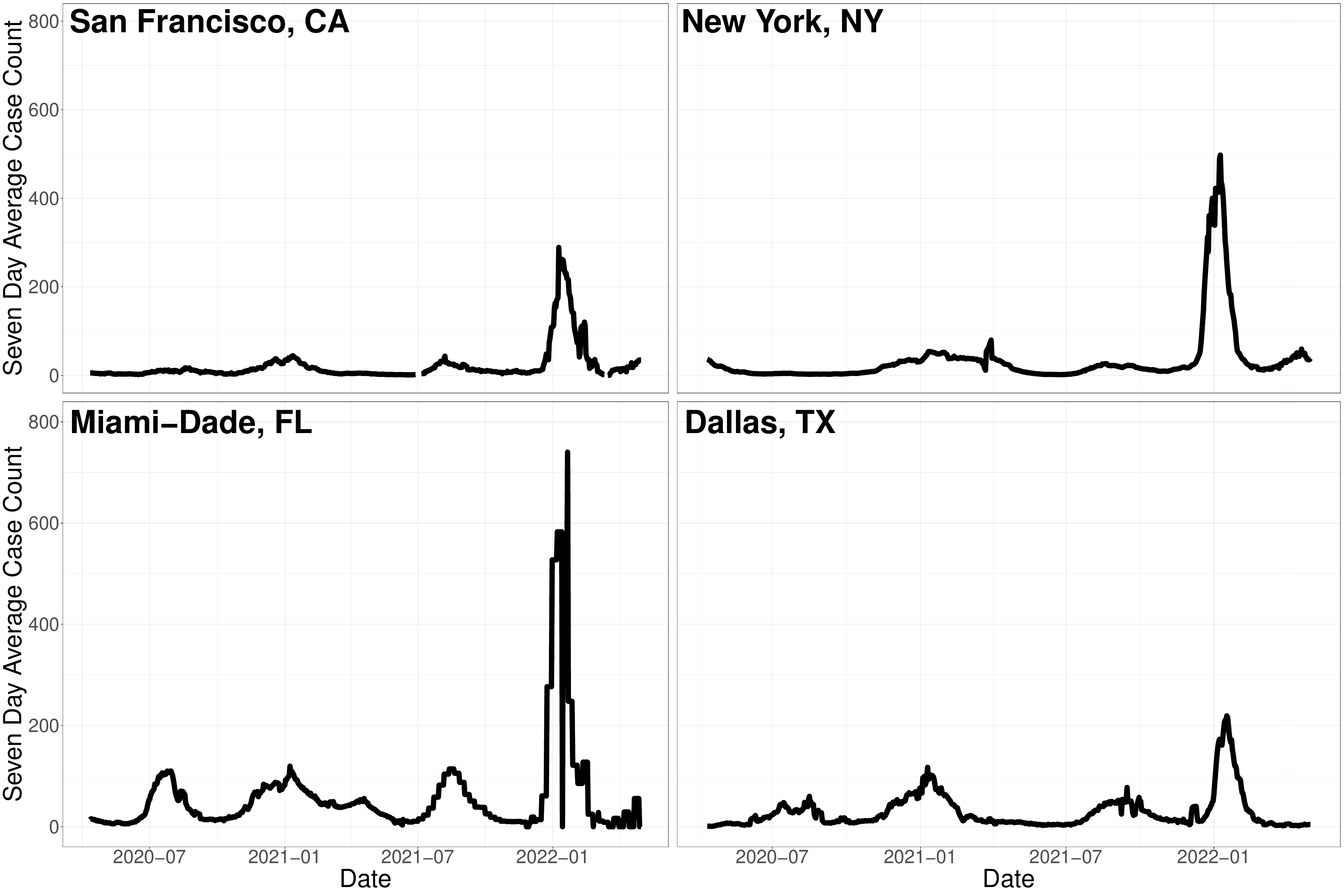}
\caption{Moving seven day averages of the number of new confirmed COVID-19 cases per 100,000 people in the four counties considered in Section \ref{sec:covid}.  }
\label{fig:covid_case_counts}
\end{centering}
\end{figure} 

\end{document}